\documentclass[12pt,a4paper]{article}
\usepackage[english]{babel}
\usepackage{float}
\usepackage{amsmath, amssymb, amsthm, epsfig, color, mathtools, bbm, dsfont, mathrsfs, bm}
\usepackage[T1]{fontenc}
\usepackage{indentfirst}
\usepackage[shortlabels]{enumitem}
\usepackage[title]{appendix}
\usepackage{lmodern}       
\usepackage{setspace}
\usepackage{soul}
\usepackage{accents}
\usepackage{hyperref} 
\usepackage{geometry}
\usepackage{xcolor}

\usepackage{xr}    
\usepackage{graphicx} 
\usepackage[export]{adjustbox} 
\usepackage{subcaption} 
\usepackage{siunitx}
\usepackage{bm} 
\usepackage{comment}

\usepackage{physics}
\usepackage{amsmath}
\usepackage{tikz}
\usepackage{mathdots}
\usepackage{yhmath}
\usepackage{cancel}
\usepackage{color}
\usepackage{siunitx}
\usepackage{array}
\usepackage{multirow}
\usepackage{amssymb}
\usepackage{gensymb}
\usepackage{tabularx}
\usepackage{extarrows}
\usepackage{booktabs}

\usepackage{addfont}
\addfont{OT1}{rsfs10}{\rsfs}

\usepackage{tikz, tkz-euclide}
\usetikzlibrary{fadings}
\usetikzlibrary{shadows.blur}
\usetikzlibrary{shapes}
\usetikzlibrary{patterns}
\usetikzlibrary{arrows.meta}
\usetikzlibrary{calc}
\usetikzlibrary{decorations.markings}
\usepackage{fancybox}

\newcommand{\be}{\begin{equation}}
\newcommand{\ee}{\end{equation}}

\newcommand{\ctr}{\gamma}

\newcommand{\diam}{\mathrm{diam}}
\newcommand{\dis}{\mathrm{dist}}
\newcommand{\I}{\mathrm{I}}
\newcommand{\Sp}{\mathrm{sp}}
\newcommand{\Z}{\mathbb{Z}}
\newcommand{\lab}{\mathrm{lab}}
\newcommand{\inB}{\partial_{\text{in}}}
\numberwithin{equation}{section}

\newcommand{\ctrb}{\Gamma}
\newcommand{\extb}{\Gamma^e}
\newcommand{\ctrbex}{\extb}

\newtheorem*{theorem*}        {Theorem}
\newtheorem*{conjecture*}   {Conjecture}
\newtheorem{theorem}{Theorem}[section]
\newtheorem{lemma}[theorem]{Lemma}
\newtheorem*{lemma*}          {Lemma}

\newtheorem{corollary}[theorem]{Corollary}
\newtheorem{proposition}[theorem]{Proposition}

\newtheorem{remark}          {Remark}[section]

\theoremstyle{definition}
\newtheorem{definition}[theorem]{Definition}

\makeatletter

\makeatletter
\def\moverlay{\mathpalette\mov@rlay}
\def\mov@rlay#1#2{\leavevmode\vtop{%
		\baselineskip\z@skip \lineskiplimit-\maxdimen
		\ialign{\hfil$\m@th#1##$\hfil\cr#2\crcr}}}
\newcommand{\charfusion}[3][\mathord]{
	#1{\ifx#1\mathop\vphantom{#2}\fi
		\mathpalette\mov@rlay{#2\cr#3}
	}
	\ifx#1\mathop\expandafter\displaylimits\fi}
\makeatother

\newcommand{\cupdot}{\charfusion[\mathbin]{\cup}{\cdot}}
\newcommand{\bigcupdot}{\charfusion[\mathop]{\bigcup}{\cdot}}

\newgeometry{vmargin={15mm}, hmargin={12mm,17mm}}

\allowdisplaybreaks
\begin{document}

\begin{center}
{\LARGE Phase Transition in Long-Range $q-$state Models via Contours. Clock and Potts models with Fields.}
\vskip.5cm
Lucas Affonso, Rodrigo Bissacot, Gilberto Faria, Kelvyn Welsch
\vskip.3cm
\begin{footnotesize}
Institute of Mathematics and Statistics (IME-USP), University of S\~{a}o Paulo, Brazil\\
\end{footnotesize}
\vskip.1cm
\begin{scriptsize}
emails: lucas.affonso.pereira@gmail.com, rodrigo.bissacot@gmail.com, gilberto.araujo@ifmt.edu.br, kelvyn.emanuel@gmail.com
\end{scriptsize}

\end{center}

\begin{abstract}
    Using the group structure of the state space of $q-$state models, a new definition of contour for long-range spin-systems in $\Z^d$ ($d\geq 2$), and a multidimensional version of Fr\"{o}hlich-Spencer contours, we prove phase transition for a class of ferromagnetic long-range systems which includes the Clock and Potts models. Our arguments work for the entire region of exponents of regular power-law interactions, namely $\alpha > d$, and for any $q \geq 2$. As an application, we prove phase transition for Potts models with decaying fields when the field decays fast enough and in the presence of a random external field. 
\end{abstract}

\section{Introduction}

After the Ising model \cite{ising1925}, one of the most studied models in statistical mechanics is its natural generalization when we have a $q$-state space ($q\geq 2$), the Potts model \cite{Potts1952} (for applications in several different areas of science see \cite{Rozikov2021}). Since its appearance, a good amount of the literature was produced about the Potts model (we will mention a non-exhaustive list of papers), using several different tools like reflection positivity \cite{Biskup2003, Biskup2006, Biskup2006_rp, Koteck1982, vanEnter1995}, mean-field theory \cite{Biskup2003, Biskup2006, Gobron2007, Pearce1980}, random-cluster model  \cite{Aizenman1988, Beffara2011, Biskup2000, bjornberg, Cioletti2015, Coquille2013, DUMINILCOPIN2021, DuminilCopin2016, Ray2020}, and contours \cite{Park.88.II, Pirogov.Sinai.75,Zahradnik.84}. Many of the results have as their primary goal the description of the Gibbs measures at low temperatures and at the critical temperature, but, additionally, they have to put further restrictions, such as assuming that the dimension is $d=2$, that the number of states $q$ is big enough (with respect to the dimension $d$), or that the dimension $d$ is sufficiently large. Most of the results consider nearest-neighbor interactions. When long-range interactions are considered, in the case of power-law decay, they do not cover the entire region of the exponents of regular interactions \cite{Biskup2006, Park.88.I}.

Since the emergence of Peierls' argument \cite{Peierls.1936}, contours have proven to be one of the most useful tools to get information about lattice systems at low temperature, culminating in the celebrated Pirogov-Sinai theory \cite{Pirogov.Sinai.75,Zahradnik.84}. In recent years, contour-based techniques in statistical mechanics and disordered systems have gained fresh impetus \cite{Ding2023}. Nonetheless, the dependence between different spins in long-range systems has a much more complex and rich structure, while the usual notion of connected contours has limited power to treat them due to the difficulty to control the interaction between these contours \cite{Park.88.II}.

In this paper, we will use the contours defined in \cite{Johanes} (see also \cite{cluster}), which were inspired by the generalization of the one-dimension contours introduced by Fröhlich and Spencer \cite{Frohlich.Spencer.82} to dimension $d \geq 2$ performed in \cite{Affonso.2021}. Adopting these contours, the aforementioned control of the interactions between them is feasible (see section \ref{energy}). With such control, we can study the phase transition phenomenon for long-range lattice models over a finite state space with mild restrictions. Our strategy combines our new definition of contour for long-range systems with an old approach which considers the group structure of the state space $\Z_q = \{1, \dots, q\}$ as in Ginibre \cite{Ginibre1970}, and also in Gruber, Hintermann and Merlini  \cite{Gruber}. The formalism allows us to employ the theory of Fourier analysis on finite groups and deal with a large class of interactions, including the Potts and Clock models.

\underline{\textit{Clock model}}: Also known as the vector Potts model, the clock model was introduced by Renfrey Potts in his PhD thesis \cite{Potts1952}, based on a suggestion by his advisor, Cyril Domb. The model generalizes the Ising model to describe situations where spins are not confined to a single direction but instead the $q$ states are uniformly distributed over the circle $S^1$. The formal Hamiltonian is given by

\begin{equation}
    H = - \sum_{x, y} J_{xy} \cos\left(\frac{2\pi}{q}(\sigma_x - \sigma_y)\right).
\end{equation}

\vspace{0.4cm}

\underline{\textit{Potts model}}: Also appeared for the first time in \cite{Potts1952}. The formal Hamiltonian is given by

\begin{equation}
    H = -\sum_{x, y} J_{xy} \mathbbm{1}_{\{\sigma_x = \sigma_y\}}.
\end{equation}

In this paper, we are only concerned with the symmetry-breaking phase transition, that occurs in low temperature. Such phase transition is characteristic of, for example, the Ising model, which corresponds to the case $q = 2$. Nevertheless, the behavior can change drastically in an arbitrary $q-$state Potts model. When $q$ is large enough, there is another kind of phase transition, which is a first-order transition in the temperature --- the $q$ ordered phases coexist with a disordered one. This was first proved by \cite{Koteck1982}, using reflection-positivity and was also accomplished by \cite{Bricmont1985} by means of contours and a refined version of Pirogov-Sinai theory. This version is an adaptation of \cite{Pirogov.Sinai.75, Pirogov1976, Sinai2014-pe} in which the ground states are replaced by a more general object, named \emph{restricted ensembles}. Other references tackling this type of phase transition are \cite{Biskup2006, DUMINILCOPIN2021, DuminilCopin2016, Laanait1991, Martirosian1986, Ray2020}.

Most of the results concerning the Potts and Clock models are restricted to short-range interactions --- that is, to cases where there exists $R > 0$ such that $J_{xy} = 0$ if $|x - y|>R$. Some exceptions are \cite{Aizenman1988, Biskup2006, Imbrie.Newman.88, Park.88.I, Park.88.II}.

In fact, the methods presented in this paper allow us to prove phase transition for any model in $\mathbb{Z}^d$ with $d \geq 2$ whose formal Hamiltonian can be written as 

\begin{equation}\label{Potts_sem_campo}
    H(\sigma) = - \sum_{x, y} J_{xy} \varphi(\sigma_x - \sigma_y),
\end{equation}
where $\varphi: \Z_q \to \mathbb{R}$ is any function such that $\varphi(0) > \varphi(n), \forall n \neq 0$ (ferromagnetism), and $J_{xy}\geq 0$ decaying polynomially with \emph{any} exponent $\alpha > d$. 

The phase transition results for these models are stated as follows.

\begin{theorem}\label{thm:phase_transition}
    Let $q \geq 2$ be a natural number. Consider the Hamiltonian 
    \begin{equation}
        H_{\Lambda}^\eta(\sigma) = -\sum_{\substack{\{x, y\} \subset \Lambda}}\hspace{-0.25cm} J_{xy}\varphi(\sigma_x - \sigma_y) - \sum_{\substack{x \in \Lambda \\ y \notin \Lambda}} J_{xy} \varphi(\sigma_x - \eta_y)
    \end{equation}
    defined on the configuration space $\{1, ..., q\}^{\Z^d}$. As above, $\varphi: \Z_q \to \mathbb{R}$ is such that $\varphi(0) > \varphi(n), \forall n \neq 0$. The interactions are given by 

     \begin{equation}\label{Long-Range Interaction}
    J_{xy} \coloneqq \begin{dcases}
                   \frac{J}{|x-y|^\alpha} &\text{ if }x\neq y,\\[0.2cm]
                   0                &\text{otherwise}, 
              \end{dcases}
\end{equation}
    for any $\alpha > d$ and $J>0$. Then, for every $C \in [0, 1)$, there is $\beta_0 = \beta_0(C, \alpha, d, \varphi, q, J)$ such that the finite-volume Gibbs measure defined by Equation \eqref{PM} satisfies
    \begin{equation}\label{eq_thm1}
        \mu_{\Lambda, \beta}^r(\sigma_0 = r) > C, \quad \forall \beta > \beta_0, \quad \forall r \in \Z_q.
    \end{equation}
    
\end{theorem}
\begin{corollary}\label{corol_phase_transition}
    Suppose that the Fourier transform $\widehat{\varphi}$ is non-negative. Then, for every $r, \ell \in \{1, ..., q\}$, $r \neq \ell$ implies that the thermodynamic limits $\mu^r_\beta$ and $\mu^\ell_\beta$ obtained with monochromatic boundary conditions do exist and are different for every $\beta > \beta_0$.
\end{corollary}

Although it is possible to deduce Theorem \ref{thm:phase_transition} using information about the short-range case and correlation inequalities (like Griffiths' Inequalities presented in section \ref{prem}), we adopted a direct strategy to show the existence of phase transition by means of contours and the Peierls' argument. It is undeniable that using contours brings many advantages, providing much information about the system, such as the typical configurations. Another advantage of this approach, explored in section $6$, is the stability with respect to perturbations, like external fields, which cannot be completely studied by correlation inequalities. 

\underline{\textit{External Fields.}} For the Ising model, it is well-known by Lee-Yang theory \cite{Lee-Yang.II.1952}, that phase transition is destroyed by any non-zero uniform field, no matter how small its strength. For $q > 2$, it is instructive to consider external fields affecting each color in a distinct way. The full Hamiltonian then reads 

\begin{equation}\label{potts_com_campo}
    H^q_{\Lambda, h}(\sigma) = -\sum_{\{x, y\} \subset \Lambda} J_{xy}\mathbbm{1}_{\{\sigma_x = \sigma_y\}} - \sum_{\substack{x \in \Lambda \\ y \in \Lambda^c}} J_{xy}\mathbbm{1}_{\{\sigma_x = q\}} - \sum_{x \in \Lambda} h_{x, \sigma_x},
\end{equation}
where $h = (h_{x, r})_{\substack{x \in \Z^d, r \in \mathbb{Z}_q}}$ is any family of real numbers.

For the ferromagnetic short-range case, Pirogov-Sinai tells us that when $\beta$ is large enough, the phase diagram mimics the one for $\beta = \infty$. For example, if $h_{x, \sigma} = \lambda \delta_{\sigma, 1}$ for some $\lambda \neq 0$, then the number of extremal translation-invariant measures depends on the sign of $\lambda$. If $\lambda > 0$, there is only one measure in the thermodynamic limit, which gives a high probability to the event $\sigma_0 = 1$. If $\lambda < 0$, there is the coexistence of $q - 1$ extremal measures, each one giving a high probability to the event $\sigma_0 = r$, $r = 2, .., q$.

The situation is much more complex when $q$ is large enough and $\beta$ is near the critical value $\beta_c$.  As already said, there is the coexistence of ordered and disordered phases at $\beta_c$. It is expected that an external field does not destroy the disordered phase, giving origin to a line of coexistence in the $\beta - h$ plane \cite{Biskup2000_PRL, Biskup2000,Goldschmidt1981}. This can be proven using Pirogov-Sinai theory \cite{menino_maluqinho} or chessboard estimates (see section 4 of \cite{vanEnter1995}). The coexistence between ordered and disordered phases is known to exist not only when $q$ is large enough but also for any $q$ when $d$ is large enough (see \cite{Biskup2003}) or when the interactions are sufficiently \emph{smeared out}. When $d = 3$, this already happens for finite-range interactions \cite{Gobron2007}. For $d = 2$, to our best knowledge, one needs to ask a polynomial decay with $2 < \alpha < 4$, see \cite{Biskup2006}. Although it is a common belief among some experts that this coexistence already happens for $d = 3$, $q = 3$, and nearest-neighbors interactions, we are not aware of any rigorous results.

In the case of a non-translation invariant field, some results are known for the Ising model \cite{velenik_field, Bis2, Bis1, Cioletti2015, Pfister1991LargeDA}. In the case of a decaying field, the modification in the Hamiltonian does not change the free energy since the graph $\Z^d$ is amenable. This class of fields was introduced for the Ising model ($q = 2$) in \cite{Bis2}, a collection of results for decaying fields in $\Z^d$ is \cite{Affonso_Bissacot_Corsini_Welsch_2024, Affonso.2021, Bis2, Bis1, Bissacot.Endo.18}. There are some papers on trees with fields as well, see \cite{BEE, Bogachev2019, georgii.gibbs.measures}.

To show the robustness of our methods, we prove phase transition for the ferromagnetic Potts model with random and deterministic decaying fields. First, we consider an external field with a sufficiently fast decay, both in the long-range (Theorem~\ref{thm:potts_decaying_long}) and in the short-range case (Corollary \ref{thm:Potts_decaying_short}). The proof produces the same region of exponents as in \cite{Affonso.2021}, but we use the contours defined in \cite{Johanes} to prove the following results.

\begin{theorem}\label{thm:potts_decaying_long}
    Suppose that there is $h^\ast \geq 0$ and $\delta > (\alpha - d) \wedge 1$ such that 
    \begin{equation}\label{field_hypothesis}
        h_{x, n} \leq \frac{h^\ast}{|x|^\delta}, \forall x \in \Z^d, n \in \Z_q. 
    \end{equation}
    Then, there is phase transition (see definition \ref{def:phase_transition}) for the Hamiltonian \eqref{potts_com_campo}, when $\beta>0$ is large enough. If $\delta = (\alpha - d) \wedge 1$, there is phase transition if $h^\ast$ is small enough, and $\beta>0$ is large enough.
\end{theorem}
\begin{corollary}\label{thm:Potts_decaying_short}
    Consider the Hamiltonian \eqref{potts_com_campo} with short-range interactions given by

\begin{equation*}
    J_{xy} = 
\begin{cases}
    J & \text{ if } |x - y| = 1,\\[0.2cm]
    0 & \text{ otherwise. }
\end{cases}
\end{equation*}
As always, $J > 0$. Suppose, again, that there is $h^\ast \geq 0$ and $\delta > 1$ such that 
    \begin{equation}
        h_{x, n} \leq \frac{h^\ast}{|x|^\delta}, \forall x \in \Z^d, n \in \Z_q. 
    \end{equation}
Then, there is phase transition when $\beta>0$ is large enough. If $\delta = 1$, there is a phase transition when $h^\ast$ is small enough, and $\beta>0$ is large enough.
\end{corollary}

\begin{figure}
    \centering
 
\tikzset{
pattern size/.store in=\mcSize, 
pattern size = 5pt,
pattern thickness/.store in=\mcThickness, 
pattern thickness = 0.3pt,
pattern radius/.store in=\mcRadius, 
pattern radius = 1pt}
\makeatletter
\pgfutil@ifundefined{pgf@pattern@name@_83vppj9ci}{
\pgfdeclarepatternformonly[\mcThickness,\mcSize]{_83vppj9ci}
{\pgfqpoint{0pt}{0pt}}
{\pgfpoint{\mcSize+\mcThickness}{\mcSize+\mcThickness}}
{\pgfpoint{\mcSize}{\mcSize}}
{
\pgfsetcolor{\tikz@pattern@color}
\pgfsetlinewidth{\mcThickness}
\pgfpathmoveto{\pgfqpoint{0pt}{0pt}}
\pgfpathlineto{\pgfpoint{\mcSize+\mcThickness}{\mcSize+\mcThickness}}
\pgfusepath{stroke}
}}
\makeatother
\tikzset{every picture/.style={line width=0.75pt}} 

\begin{tikzpicture}[x=0.75pt,y=0.75pt,yscale=-1,xscale=1]

\draw  [draw opacity=0][fill={rgb, 255:red, 128; green, 128; blue, 128 }  ,fill opacity=0.71 ] (159.2,62) -- (366.42,62) -- (366.42,240.07) -- (159.2,240.07) -- cycle ;
\draw  [draw opacity=0][pattern=_83vppj9ci,pattern size=6pt,pattern thickness=0.75pt,pattern radius=0pt, pattern color={rgb, 255:red, 88; green, 88; blue, 88}] (158.2,62) -- (365.42,62) -- (365.42,240.07) -- (158.2,240.07) -- cycle ;
\draw  [draw opacity=0][fill={rgb, 255:red, 255; green, 255; blue, 255 }  ,fill opacity=1 ] (238.79,159) -- (366.67,159) -- (366.67,240.07) -- (238.79,240.07) -- cycle ;
\draw  [draw opacity=0][fill={rgb, 255:red, 255; green, 255; blue, 255 }  ,fill opacity=1 ] (158.2,240.07) -- (239.79,158.92) -- (239.79,240.07) -- cycle ;
\draw    (159.2,250.47) -- (159.89,44.97) ;
\draw [shift={(159.9,42.97)}, rotate = 90.19] [color={rgb, 255:red, 0; green, 0; blue, 0 }  ][line width=0.75]    (10.93,-3.29) .. controls (6.95,-1.4) and (3.31,-0.3) .. (0,0) .. controls (3.31,0.3) and (6.95,1.4) .. (10.93,3.29)   ;
\draw  [dash pattern={on 4.5pt off 4.5pt}]  (159.2,240.07) -- (239.79,158.92) ;
\draw  [dash pattern={on 4.5pt off 4.5pt}]  (240.83,159.42) -- (435.17,158.73) ;
\draw [shift={(438.17,158.72)}, rotate = 179.8] [fill={rgb, 255:red, 0; green, 0; blue, 0 }  ][line width=0.08]  [draw opacity=0] (10.72,-5.15) -- (0,0) -- (10.72,5.15) -- (7.12,0) -- cycle    ;
\draw    (162,160.71) -- (157,160.71) ;
\draw    (240.5,236.92) -- (240.5,243.46) ;
\draw    (148.6,240.6) -- (424.1,240.1) ;
\draw [shift={(426.1,240.1)}, rotate = 179.9] [color={rgb, 255:red, 0; green, 0; blue, 0 }  ][line width=0.75]    (10.93,-3.29) .. controls (6.95,-1.4) and (3.31,-0.3) .. (0,0) .. controls (3.31,0.3) and (6.95,1.4) .. (10.93,3.29)   ;
\draw  [fill={rgb, 255:red, 255; green, 255; blue, 255 }  ,fill opacity=1 ] (204,105.14) .. controls (204,102.49) and (206.15,100.33) .. (208.81,100.33) -- (321.91,100.33) .. controls (324.57,100.33) and (326.72,102.49) .. (326.72,105.14) -- (326.72,119.58) .. controls (326.72,122.23) and (324.57,124.39) .. (321.91,124.39) -- (208.81,124.39) .. controls (206.15,124.39) and (204,122.23) .. (204,119.58) -- cycle ;
\draw  [fill={rgb, 255:red, 255; green, 255; blue, 255 }  ,fill opacity=1 ] (269.39,190.14) .. controls (269.39,188.22) and (270.95,186.67) .. (272.87,186.67) -- (353.58,186.67) .. controls (355.5,186.67) and (357.06,188.22) .. (357.06,190.14) -- (357.06,200.58) .. controls (357.06,202.5) and (355.5,204.06) .. (353.58,204.06) -- (272.87,204.06) .. controls (270.95,204.06) and (269.39,202.5) .. (269.39,200.58) -- cycle ;
\draw   (444.5,146.54) .. controls (444.5,142.38) and (447.88,139) .. (452.04,139) -- (547.62,139) .. controls (551.79,139) and (555.17,142.38) .. (555.17,146.54) -- (555.17,169.18) .. controls (555.17,173.34) and (551.79,176.72) .. (547.62,176.72) -- (452.04,176.72) .. controls (447.88,176.72) and (444.5,173.34) .. (444.5,169.18) -- cycle ;
\draw    (366.3,237.62) -- (366.3,244.16) ;

\draw (420.8,243.97) node [anchor=north west][inner sep=0.75pt]    {$\alpha - d$};
\draw (138.5,50.48) node [anchor=north west][inner sep=0.75pt]    {$\delta $};
\draw (235.08,247.61) node [anchor=north west][inner sep=0.75pt]    {$1$};
\draw (141.33,151.69) node [anchor=north west][inner sep=0.75pt]    {$1$};
\draw (210.81,105.33) node [anchor=north west][inner sep=0.75pt]  [font=\small] [align=left] {Phase Transition};
\draw (275.33,189.33) node [anchor=north west][inner sep=0.75pt]   [align=left] {{\footnotesize Uniqueness?}};
\draw (499.69,158.61) node  [font=\footnotesize] [align=left] {Phase Transition\\ \ \ \ \ for small $\displaystyle h^{\ast }$};
\draw (358.8,245.47) node [anchor=north west][inner sep=0.75pt]    {$\infty $};
\draw (327.67,259.67) node [anchor=north west][inner sep=0.75pt]  [font=\tiny] [align=left] {\textit{(nearest-neighbors)}};

\end{tikzpicture}
    \caption{Phase Diagram.}
    \label{phase_diagram}
\end{figure}

The random field long-range Potts model is defined as the system with Hamiltonian \eqref{potts_com_campo}, where the external field is a family $\{ \varepsilon h_{x,n}\}_{x\in\Z^d,  n\in\Z_q}$ of i.i.d. random variables instead of real numbers. Each $h_{x, n}$ has a standard normal distribution. This is a generalization for general $q$ of the random field long-range Ising model, studied in \cite{Johanes}, where phase transition was proved for all $\alpha>d$ and $d\geq 3$. The argument was based on the recent new proof of Ding and Zhuang \cite{Ding2023} of the corresponding result for the random field nearest neighbor Ising model, based on a Peierls argument. Their argument greatly simplifies the previously available result of Bricmont and Kupiainen \cite{Bricmont.Kupiainen.88}, which uses the Renormalization Group Method. Ding and Zhuang also showed that phase transition holds for the corresponding random field Potts model. With our results, we can prove the following result:

\begin{theorem}\label{thm:random_field_potts}
    Given $d,q\geq 3$, $\alpha>d$, there exists $\beta_c\coloneqq\beta(d, \alpha, q)$ and $\varepsilon_c\coloneqq\varepsilon(d, \alpha, q)$ such that, for $\beta> \beta_c$ and $\varepsilon\leq \varepsilon_c$ the long-range random field Potts model presents phase transition $\mathbb{P}$-almost surely.
\end{theorem}


 This paper is divided as follows. In section \ref{prem}, we present the relevant definitions. We also revisit correlation inequalities and the thermodynamic limit for $q$-state spin systems. The new contours are the protagonist of section \ref{contour_section}, where the exponential growth in the number of possible contours is an important feature and can be found in \cite{Johanes}. The main computation is the energetic bound presented in section \ref{energy}. These two ingredients are combined in section \ref{pt}, which consists of the proof of Theorem \ref{thm:phase_transition}. The applications for models with decaying and random fields are proved in section \ref{application}. We finished the paper with section \ref{concluding}, where we mention possible consequences and problems for which this new notion of multi-scaled disconnected contours can be useful.

\section{Preliminaries}\label{prem}

Given $\Lambda \subset \mathbb{Z}^d$, we define the local configuration space as $\Omega_\Lambda := (\Z_q)^\Lambda$. When $\Lambda = \mathbb{Z}^d$, we simply put $\Omega := (\Z_q)^{\mathbb{Z}^d}$. Fixed $\eta \in \Omega$, we also define $\Omega_{\Lambda}^\eta$ as the subset of $\Omega$ consisting of configurations such that $\sigma_x = \eta_x $ for each $x \notin \Lambda$. Finally, we write $\Lambda \Subset \mathbb{Z}^d$ to indicate that $\Lambda$ is finite. 

 Given $\eta \in \Omega$ and $\Lambda \Subset \Z^d$, we will be interested in models whose Hamiltonian can be written as follows:

\begin{equation}\label{hamiltonian}
H_{\Lambda, h}^\eta(\sigma) = -\sum_{\substack{\{x, y\} \subset \Lambda}}J_{xy}\varphi(\sigma_x - \sigma_y) - \sum_{\substack{x \in \Lambda \\ y \notin \Lambda}} J_{xy} \varphi(\sigma_x - \eta_y) - \sum_{x \in \Lambda} h_x \varphi(\sigma_x),
\end{equation}
where $h = (h_x)_{x \in \Z^d}$ is any family of real numbers and  $\{J_{xy}\}_{x,y\in\Z^d}$ is defined by Equation \eqref{Long-Range Interaction}
for some $J > 0$ and $\alpha > d$.  Furthermore, we ask the function $\varphi: \Z_q \to \mathbb{R}$ to be such that $\varphi(0) > \varphi(n), \forall n \neq 0$ (ferromagnetism). We will restrict our attention to monochromatic boundary conditions, that is, when $\eta_x = r, \forall x \in \mathbb{Z}^d$, for some $r \in \Z_q$, in which case we will simply write $H^r_{\Lambda, h}$.

Denote by $\mathcal{F}_\Lambda$ the $\sigma-$algebra generated by the cylindrical sets supported on $\Lambda$ and write $\mathcal{F}=\mathcal{F}_{\mathbb{Z}^d}$.

\begin{definition}
For any $\Lambda \Subset \mathbb{Z}^d$, $\eta \in \Omega$ and $\beta > 0$, we define the corresponding finite-volume Gibbs measure on $(\Omega, \mathcal{F})$ by
 \begin{equation}\label{PM}
        \mu_{\Lambda, \beta, h}^\eta(\sigma) \coloneqq \mathbbm{1}_{\Omega_\Lambda^\eta}(\sigma)\frac{e^{-\beta H_{\Lambda,h}^{\eta}(\sigma)}}{Z_{\Lambda, \beta}^{\eta}(h)},
    \end{equation}
where $\beta$ has the physical meaning of the inverse temperature and the normalization factor $Z_{\Lambda, \beta}^\eta(h)$ is known as the partition function, defined by
$$
Z_{\Lambda, \beta}^\eta(h):=\sum_{\sigma\in\Omega_\Lambda} e^{-\beta H_{\Lambda, h}^\eta(\sigma)}.
$$

\end{definition}


Similarly to the Hamiltonian, we write $Z^r_{\Lambda, \beta}(h)$ and $\mu^r_{\Lambda, \beta, h}$ for monochromatic boundary conditions. Moreover, when $h \equiv 0$, we will omit the subscript $h$. Notice that the collection of all finite subsets of $\mathbb{Z}^d$, $\mathcal{P}_f(\mathbb{Z}^d)$, has the structure of a directed set given by the inclusion.

\begin{definition}\label{def:phase_transition}
    Fixed $\beta > 0$ and $r \in \Omega_0$, the limit points of the net of the finite-volume Gibbs measures $(\mu_{\Lambda, \beta, h}^\eta)_{\Lambda \in \mathcal{P}_f(\mathbb{Z}^d)}$, with respect to the weak-$\ast$ topology, are called the \emph{thermodynamic limits}. The set of all thermodynamic limits for all possible boundary conditions $\eta$ will be denoted by $\mathscr{G}_{\beta, h}$. We say that a model undergoes phase transition when $|\mathscr{G}_{\beta, h}| > 1$.
\end{definition}


Since the set of all probability measures in this case is compact, there exists some thermodynamic limit Gibbs measure for any $\beta > 0$ and $r \in \Z_q$. As a consequence of theorem 1.1, we know that limit points for different monochromatic boundary conditions must be different. In itself, this result already implies the existence of (at least) $q$ different Gibbs measure. In some cases, however, it is possible to know uniqueness of the limit points for each boundary condition. For the Potts model, this statement was proven in \cite{Aizenman1988} (see also \cite{bjornberg}) using the representation in terms of the random-cluster model. A more general approach is to use the Griffiths inequalities, in the framework provided by Ginibre \cite{Ginibre1970}, this will be the subject until the end of this section.
Before presenting the result, let's introduce some notation. Denote by $\mathcal{C}(\Omega_\Lambda)$ the set of all complex continuous functions on $\Omega_\Lambda$. 

\begin{definition}[Convex Cone] Let $Q$ be a subset of a vector space $V$. The set $Q$ is called a convex cone if, for every $v_1, v_2 \in V$ and every scalars $\lambda_1, \lambda_2 \geq 0$, we have $\lambda_1 v_1 + \lambda_2 v_2 \in Q$.
\end{definition}

\begin{remark}
    In what follows, we are going to use some basic facts about harmonic analysis on locally compact Abelian groups. We refer the reader to \cite{folland} for a good exposition on the subject.
\end{remark}

\begin{definition}[Positive Semi-Definite Function]
    Given a group $G$, a function $\varphi: G \to \mathbb{R}$ is said positive semi-definite if, for any finite family $g_1, ..., g_n \in G$, the matrix $(\varphi(g^{-1}_i g_j))_{ij}$ is positive semi-definite, that is, denoting by $B$ the corresponding bilinear form, then $B(v, v) \geq 0$, for any $v \in \mathbb{R}^n$.
\end{definition}

Given $S \subset \mathcal{C}(\Omega_\Lambda)$, denote by $Q(S)$ the closure of the intersection of all convex cones in $\mathcal{C}(\Omega_\Lambda)$ containing $S$ and closed under multiplication. Given $H \in \mathcal{C}(\Omega_\Lambda)$ real, define

\begin{equation*}
    \langle f \rangle_H := \left[\sum_{\omega \in \Omega_\Lambda} e^{-\beta H(\omega)} \right]^{-1}\sum_{\omega \in \Omega_\Lambda} f(\omega) e^{-\beta H(\omega)}.
\end{equation*}

\begin{theorem}[Ginibre, 1970 \cite{Ginibre1970}]\label{ginibre} Let $S \subset \mathcal{C}(\Omega_\Lambda)$ be a self-conjugate set and $-H \in Q(S)$. If, for any finite collection $f_1, ..., f_n \in S$ and any finite sequence $s_1, ..., s_n \in \{0, 1\}$,

\begin{equation}\label{q3}
    \sum_{\sigma \in \Omega_\Lambda} \sum_{\omega \in \Omega_\Lambda} \prod_{i = 1}^n \left(f_i(\sigma) + (-1)^{s_i}f_i(\omega) \right) \geq 0,
\end{equation}

then the two Griffiths' inequalities hold. That is, 

\begin{enumerate}
    \item $\langle f \rangle_H \geq 0$, $\forall f \in Q(S)$,
    \item $\langle fg \rangle_H - \langle f \rangle_H \langle g \rangle_H \geq 0$, $\forall f, g \in Q(S)$.
\end{enumerate}
\end{theorem}

The condition \eqref{q3} is called $(Q3)$ in \cite{Ginibre1970}. By example $4$ of \cite{Ginibre1970}, \eqref{q3} holds if we take $S = S_\Lambda$ as the set of real positive semi-definite functions in $\Omega_\Lambda$. Since $S_\Lambda = Q(S_\Lambda)$, the Theorem above tells us that the Griffiths' Inequalities hold provided that $-H$ is positive semi-definite. The following lemma gives us another characterization for positive semi-definite functions.
    
\begin{lemma}
   Let $\varphi: G \to \mathbb{C}$ be a function in $L^1(G)$. If the Fourier Transform $\widehat{\varphi}$ is in $L^1(\widehat{G})$ and $\widehat{\varphi} \geq 0$, then $\varphi$ is positive semi-definite.
\end{lemma}
\begin{proof}
    Let $g_1, ..., g_n$ be a finite collection of elements in $G$. We want to show that the matrix $(\varphi(g^{-1}_i g_j))_{ij}$ is positive semi-definite. Since $\widehat{\varphi} \in L^1(\widehat{G})$, we can use the Fourier inversion formula, which tells us that

    \begin{equation*}
        \varphi(g^{-1}_i g_j) = \int_{\widehat{G}} \widehat{\varphi}(\chi) \widehat{g^{-1}_i g_j}(\chi) d\widehat{\mu}(\chi) = \int_{\widehat{G}} \widehat{\varphi}(\chi) \widehat{g^{-1}_i}(\chi) \widehat{g_j}(\chi) d\widehat{\mu}(\chi),
    \end{equation*}
where $\widehat{\mu}$ is the Pontryagin dual measure of some Haar measure $\mu$ in $G$, and $\widehat{g}(\chi) = \chi(g)$ is the evaluation map. Now, notice that the bilinear form $\langle \cdot, \cdot \rangle_{h}$ defined on $\text{span}\{\widehat{g_1}, ..., \widehat{g_n}\}$ by 

\begin{equation*}
    \langle u, v \rangle_h = \int_{\widehat{G}} h(\chi) u(\chi) \overline{v(\chi)} d\widehat{\mu}(\chi)
\end{equation*}
is positive semi-definite provided that $h(\chi) \geq 0$, $\forall \chi \in \widehat{G}$. Recalling that a matrix given by $(\langle v_i, v_j\rangle)_{ij}$ is positive semi-definite if $\langle \cdot, \cdot \rangle$ is so, we have that the matrix $(\varphi(g^{-1}_i g_j))_{ij} = \langle \widehat{g_i}, \widehat{g_j}\rangle_{\widehat{\varphi}})_{ij}$ is positive semi-definite.

\end{proof}

\begin{remark}
    By the Bochner's Theorem (see Theorem 4.18 of \cite{folland}), a function is positive semi-definite if, and only if, its Fourier transform is non-negative.
\end{remark}

\begin{proposition}
    The Fourier transforms of the functions $\varphi_{\text{cl}}(n) = \cos(2\pi n/q)$ and $\varphi_{\text{p}}(n) = \mathbbm{1}_{\{n = 0\}}$ are non-negative.
\end{proposition}
\begin{proof}
Recall that every character of $\mathbb{Z}_q$ can be written in the form $\chi_k(n) = \exp\left(i \frac{2\pi}{q}kn \right)$, for some $k \in \mathbb{Z}_q$. For $\varphi_{\text{cl}}$, we can write

\begin{equation*}
    \varphi_{\text{cl}}(n) = \cos\left(\frac{2\pi}{q}n\right) = \frac{1}{2}\exp\left(i\frac{2\pi}{q}n\right) + \frac{1}{2}\exp\left(-i\frac{2\pi}{q}n\right) = \frac{1}{2} \chi_{1}(n) + \frac{1}{2}\chi_{q-1}(n).
\end{equation*}

Since the Fourier Transform must be proportional to these coefficients, we conclude that it must be non-negative.

For the Potts one, we have $\widehat{\varphi_{\text{p}}}(k) = \sum_{n} \varphi_{\text{p}}(n) \overline{\chi_{k}(n)} = \overline{\chi_{k}(0)} = 1$, so $\widehat{\varphi_{\text{p}}}$ is also non-negative.
\end{proof}

\begin{corollary}\label{clock}
    Let $\Lambda \Subset \mathbb{Z}^d$ and $S_{\Lambda}$ be the set of all real positive semi-definite functions on $\Omega_{\Lambda}$. Then, for any Hamiltonian of the form \eqref{hamiltonian}, if $\varphi$ is positive semi-definite, we have  
    \begin{enumerate}
    \item $\langle f \rangle_{\Lambda, \beta, h}^q \geq 0$;
    \item $\langle f g \rangle_{\Lambda, \beta, h}^q  - \langle f \rangle_{\Lambda, \beta, h}^q  \langle g \rangle_{\Lambda, \beta, h}^q  \geq 0$,
\end{enumerate}
for any $f, g \in S_{\Lambda}$.
\end{corollary}
\begin{proof}
In first place notice that, although the results in Theorem \ref{ginibre} (according to \cite{Ginibre1970}) are restricted to finite volumes only, for any $\Lambda$-local function $f: \Omega \to \mathbb{C}$, we have

\begin{equation*}
    \langle f \rangle_{\Lambda, \beta, h}^q  = \langle f \rangle_{H^q_{\Lambda, h}} :=  \left[\sum_{\omega \in \Omega_\Lambda} e^{-\beta H^q_{\Lambda, h}(\omega)} \right]^{-1}\sum_{\omega \in \Omega_\Lambda} f(\omega) e^{-\beta H^q_{\Lambda, h}(\omega)},
\end{equation*}
where on the right-hand side, both $f$ and $H^q_{\Lambda, h}$ are being regarded as functions on $\Omega_{\Lambda}$, since they are $\Lambda$-local.


Due to the previous lemma, we only need to show that $\widehat{-H_{\Lambda, h}^q}$ is non-negative. Given an Abelian and finite group $G$, define $\delta: G \times G \to G$ by $\delta(g, h) = g - h$. If $\varphi : G \to \mathbb{C}$ has a non-negative Fourier transform, then $\varphi \circ \delta: G \times G \to \mathbb{C}$ has a non-negative Fourier transform as well. Indeed, recall that the dual of the product of two groups is the product of the respective dual groups (see Proposition 4.6 of \cite{folland}). Thus,

\begin{align*}
    \widehat{\varphi \circ \delta}(\chi_1, \chi_2) &= \sum_{n_1, n_2} \varphi(n_1 - n_2) \overline{\chi_1(n_1)}\ \overline{\chi_2(n_2)}.
\end{align*}

Using the Inversion Formula,

\begin{equation*}
    \varphi(n) = \sum_{\xi \in \widehat{G}} \alpha_\xi \xi(n),
\end{equation*}

where $\alpha_\xi \geq 0$, by the hypothesis that $\varphi$ has a non-negative Fourier transform. Substituting,

\begin{align*}
    \widehat{\varphi \circ \delta}(\chi_1, \chi_2) &= \sum_{n_1, n_2} \sum_{\xi \in \widehat{G}} \alpha_\xi \xi(n_1 - n_2) \overline{\chi_1(n_1)}\ \overline{\chi_2(n_2)} \\
    & = \sum_{\xi \in \widehat{G}} \alpha_\xi \left[\sum_{n_1 \in G} \xi(n_1) \overline{\chi_1(n_1)}  \right] \left[\sum_{n_2 \in G} \xi^{-1}(n_2) \overline{\chi_2(n_2)}  \right].
\end{align*}

Recall that the characters of an Abelian finite group satisfy the following orthogonality relation:

\begin{equation*}
    \sum_{g \in G} \chi_1(g) \overline{\chi_2(g)} =
    \begin{dcases}
        |G| &\text{ if } \chi_1 = \chi_2, \\[0.2cm]
        0 &\text{ otherwise.}
    \end{dcases}
\end{equation*}

This means that the only term of the summation over $\xi$ that will be non-zero is the term $\xi = \chi_1$, provided that $\chi_1 = \chi_2^{-1}$. In summary,

\begin{equation*}
     \widehat{\varphi \circ \delta}(\chi_1, \chi_2) =
     \begin{dcases}
         \alpha_{\chi_1}|G|^2 &\text{ if } \chi_1 = \chi_2^{-1}, \\[0.2cm]
         0 &\text{ otherwise.}
     \end{dcases}   
\end{equation*}

This shows that $\widehat{\varphi \circ \delta}(\chi_1, \chi_2) \geq 0$, as desired.

Finally, we need to show that, if $\psi: G_1 \to \mathbb{C}$ has a non-negative Fourier transform, and $\pi_1(g_1, g_2) = g_1$, then $\psi \circ \pi_1: G_1 \times G_2 \to \mathbb{C}$ has a non-negative Fourier transform as well. In fact, 

\begin{align*}
    \widehat{\psi \circ \pi_1}(\chi_1, \chi_2) &= \sum_{n_1 \in G_1} \sum_{n_2 \in G_2} (\psi \circ \pi_1)(n_1, n_2) \overline{\chi_1(n_1)}\ \overline{\chi_2(n_2)}\\[0.3cm]
    &= \sum_{n_2 \in G_2} \left( \sum_{n_1 \in G_1} \psi(n_1) \overline{\chi(n_1)} \right) \overline{\chi_2(n_2)} \\[0.4cm]
    &= \widehat{\psi}(\chi_1) \sum_{n_2 \in G_2} \overline{\chi_2(n_2)}.
\end{align*}
Then, 
\begin{equation*}
     \widehat{\psi \circ \pi_1}(\chi_1, \chi_2) =
    \begin{dcases}
        |G_2| \widehat{\psi}(\chi_1) &\text{ if } \chi_2 \equiv 1, \\[0.3cm]
        0 &\text{ otherwise.}
    \end{dcases}
\end{equation*}

With the previous facts and using that the set of real positive semi-definite functions is a convex cone, we have that $J_{xy}\varphi(\sigma_x - \sigma_y)$ are positive semi-definite on $\Omega_\Lambda$, and hence the whole Hamiltonian.
\end{proof}

By standard methods we can prove the following proposition.

\begin{proposition}\label{monotone}
For any $\Lambda$-local function $f: \Omega \to \mathbb{R}$ with non-negative Fourier transform and $\beta > 0$,
\begin{enumerate}
    \item the mapping $h_z \mapsto \langle f \rangle^q_{\Lambda, \beta, h}$ is non-decreasing  for any $z \in \Lambda$.
    \item  $\langle f \rangle^q_{\Delta, \beta, h} \leq \langle f \rangle_{\Lambda, \beta, h}^q $, for any $\Lambda \subset \Delta \Subset \mathbb{Z}^d$.
\end{enumerate}
    
\end{proposition}


\begin{corollary}\label{thermo_limit}
    For any local function $f: \Omega \to \mathbb{R}$ and $r \in \Z_q$, $\lim_{\Lambda \uparrow \mathbb{Z}^d} \langle f \rangle_{\Lambda, \beta}^r  $ exists. 
\end{corollary}
\begin{proof}
Let's start supposing that $r$ is the identity $q$. The Proposition above, together with the first Griffiths' inequality shows us that $\lim_{\Lambda \uparrow \mathbb{Z}^d} \langle f \rangle_{\Lambda, \beta}^q $ must exist whenever $f$ is a local function with non-negative Fourier transform. Now, let $f$ be any real local function. If $f$ is an odd function, the fact that $H^q_{\Lambda}$ is even implies that $\langle f \rangle_{\Lambda, \beta}^q = 0$. Thus, we may suppose without loss of generality that $f$ is even. By the inversion formula, we can write $f = \sum_k a_k \chi_k$, where $(\chi_k)_k$ are the characters of $(\mathbb{Z}_q)^{\Lambda}$, for some $\Lambda$ where $f$ is local. Since $f$ is even, the coefficients are real and we can split the previous sum in its positive and negative parts. Explicitly, we define $f_+ := \sum_k b_k \chi_k$ and $f_- := \sum_k c_k \chi_k$., where $b_k := \max(a_k, 0)$ and $c_k := \max(-a_k, 0)$, so we can write $f = f_+ - f_-$ such that both $f_+$ and $f_-$ have a non-negative Fourier transform. Since $f$ is real, we know that $k \mapsto a_k$ needs to be even. By construction, it is obvious that both $k \mapsto b_k$ and $k \mapsto c_k$ are even, so $f_+$ and $f_-$ are real. By the last proposition, 
\begin{equation*}
    \lim_{\Lambda \uparrow \mathbb{Z}^d} \langle f \rangle_{\Lambda, \beta, h}^q  = \lim_{\Lambda \uparrow \mathbb{Z}^d} \langle f_+ \rangle^q_{\Lambda, \beta, h} - \lim_{\Lambda \uparrow \mathbb{Z}^d} \langle f_- \rangle^q_{\Lambda, \beta, h}
\end{equation*}

so the conclusion follows. Now, take any $r \in \Z_q$ and define $\tau_r: \Omega_\Lambda \to \Omega_\Lambda$ by $(\tau_r(\sigma))_x = \sigma_x - r$. 

 Notice that $H^r_\Lambda(\sigma) = H^q_\Lambda(\tau_r(\sigma))$. Hence,

\begin{align*}
    \langle f \rangle_{\Lambda, \beta}^r = \sum_{\sigma \in \Omega_\Lambda} f(\sigma) e^{-\beta H^r_{\Lambda}(\sigma)} = \sum_{\sigma \in \Omega_\Lambda} f(\tau_r^{-1}(\tau_r(\sigma))) e^{-\beta H^q_{\Lambda}(\tau_r(\sigma))}.
\end{align*}

 Since $\tau_r$ is a bijection, we have $\langle f \rangle_{\Lambda, \beta}^r = \langle f \circ \tau_r^{-1} \rangle_{\Lambda, \beta}^q$. By what was already proven, the limit exist for $f \circ \tau^{-1}_r$ in the $q-$boundary condition, so the limit of $f$ exists in the $r-$boundary condition.
\end{proof}

\begin{remark}(Proof of Phase Transition via Griffiths' Inequalities) If we highlight the dependence with respect to the coupling $J_{xy}$, writing $\langle f \rangle^q_{\Lambda, \beta, \mathbf{J}}$, we can prove by standard methods that $J_{xy} \mapsto \langle f \rangle^q_{\Lambda, \beta, \mathbf{J}}$ is non-decreasing for any $x, y \in \Z^d$ and any real, local function $f$ that is positive semi-definite. We know that for $d \geq 2$, the nearest-neighbors Potts model presents phase transition at low temperatures. The monotonicity with respect to $J_{xy}$ implies the phase transition for the long-range Potts model. However, our goal is to present the new contours and a direct proof of the phase transition; the approach with contours can be used for further applications as for dealing with models with decaying fields, and many other problems.
\end{remark}
\section{Contours}\label{contour_section}

In this section we define the notion of $(M,a)$-\emph{partition}, which allow us to define the analogous to the Fr{\"o}hlich-Spencer contours in the multidimensional setting.

\begin{definition}\label{def:incorrect_points}
     Given a configuration $\sigma$, a point $x \in \Z^d$ is $r$\textit{-correct} for $\sigma$ if $\sigma_y = r$  for every $y \in B_1(x)$, where $B_1(x)$ is the unit ball in the $\ell_1$-norm centered at $x \in \Z^d$. A point is called \textit{incorrect} for $\sigma$ if it's not $r-$correct for any $r \in \Z_q$. The \textit{boundary} of a configuration $\sigma$ is the set $\partial \sigma$ of all incorrect points for $\sigma$.
\end{definition}

For systems with finite-range interactions, we can define the contours of a configuration as the connected components of its boundary. In our case, the contours will also be defined by a partition of the boundary, but taking connected components is no longer suitable. We need to introduce the following notion.

\begin{definition}\label{Ma}
    Let $M>0$ and $a > d$. For each $A\Subset\Z^d$, a set $\Gamma(A) \coloneqq \{\overline{\gamma} : \overline{\gamma} \subset A\}$ is called an $(M,a)$-\emph{partition} of $A$ when the following two conditions are satisfied.
	\begin{enumerate}[label=\textbf{(\Alph*)}, series=l_after] 
		\item They form a partition of $A$, i.e.,  $\bigcupdot_{\overline{\gamma} \in \Gamma(A)}\overline{\gamma}=A$.
		
		\item For all $\overline{\gamma}, \overline{\gamma}^\prime \in \Gamma(A)$, 
			\be\label{B_distance_2}
			\dis(\overline{\gamma},\overline{\gamma}') > M\min\left \{|V(\overline{\gamma})|,|V(\overline{\gamma}')|\right\}^\frac{a}{d+1},
 			\ee
	\end{enumerate}
 where $V(\Lambda)$ denotes the \textit{volume} of $\Lambda \Subset \mathbb{Z}^d$, and is given by $V(\Lambda) \coloneqq \Z^d \setminus \Lambda^{(0)}$ with $\Lambda^{(0)}$ being the unique unbounded connected component of $\Lambda^c$. For any $A \Subset \Z^d$, we denote by $|A|$ its cardinality.
\end{definition}

Even after fixing the parameters $M$ and $a$, there can be multiple partitions of a set that are $(M, a)$-partitions. However, there is always a \textit{finest} $(M,a)$-partition and we pick this one (see \cite{Johanes} for details). The finest $(M,a)$-partition of $A\Subset \Z^d$ satisfies the following property (see \cite{cluster}):

\begin{itemize}
    \item[\textbf{(A1)}] For any $\overline{\gamma},\overline{\gamma}^\prime\in \Gamma(A)$, $\overline{\gamma}'$ is contained in only one connected component of $(\overline{\gamma})^c$.
\end{itemize}

In this paper we will use $a \coloneqq a(\alpha,d) = \frac{3(d+1)}{(\alpha-d)\wedge 1}$. The constant $M$ will be appropriately chosen later.

\begin{definition}[\textbf{Contours}]\label{def:d_contours}
Given a configuration $\sigma$ with finite boundary, its \emph{contours} $\gamma$ are pairs $(\overline{\gamma},\sigma_{\overline{\gamma}})$, where $\overline{\gamma} \in \Gamma(\partial \sigma)$. The \emph{support of the contour} $\gamma$ is defined as $\Sp(\gamma)\coloneqq \overline{\gamma}$, and its \emph{size} is given by $|\gamma| \coloneqq |\Sp(\gamma)|$. 
\end{definition}

With this definition, every configuration $\sigma \in \Omega_\Lambda^{q}$ is naturally associated to the family of contours $\Gamma(\sigma) \coloneqq \{\gamma_1, ..., \gamma_n\}$, where the respective supports are the $(M, a)$-partition of $\Gamma(\partial \sigma)$.

Given a subset $\Lambda \Subset \Z^d$ we define its \textit{interior} as $\I(\Lambda) \coloneqq V(\Lambda) \setminus \Lambda$. For the special case of a contour $\gamma$, we write $\I(\gamma)$ and $V(\gamma)$ instead of $\I(\Sp(\gamma))$ and $V(\Sp(\gamma))$. Moreover, we define $V(\Gamma) \coloneqq \bigcup_{\gamma \in \Gamma} V(\gamma)$. We also define the \emph{edge boundary} of $\Lambda$ as $\partial \Lambda := \{\{x, y\} \subset \Z^d; |x - y| = 1, x \in \Lambda, y \in \Lambda^c\}$, the inner boundary as $\partial_{in} \Lambda := \{x \in \Lambda; |x - y| = 1 \text{ for some } y \in \Lambda^c\}$ and the exterior boundary as $\partial_{ex} \Lambda := \{x \in \Lambda^c; |x - y| = 1 \text{ for some } y \in \Lambda\}$.

Also, denoting by $\I(\gamma)^{(k)}$, $k = 1, ..., n$, the connected components of $\I(\gamma)$, we can define the \emph{label} map $\lab_{\overline{\gamma}}: \{\Sp(\gamma)^{(0)}, \I(\gamma)^{(1)},\dots, \I(\gamma)^{(n)}\} \rightarrow \mathbb{Z}_q$ by taking the label of $\Sp(\gamma)^{(0)}$ as the spin of $\sigma$ in $\inB V(\gamma)$ and the label of $\I(\gamma)^{(k)}$ as the spin of $\sigma$ in $\partial_{\text{ex}} V(\I(\gamma)^{(k)})$. Notice that there can be connected components of a contour sitting inside its own interior. However, the labels are well-defined, since the spin of $\sigma$ is constant in the boundaries of $\Sp(\gamma)$. The following sets will be useful 
	\begin{equation}\label{interior}
	\I_n(\gamma) \coloneqq \hspace{-1cm}\bigcup_{\substack{k \geq 1, \\ \lab_{\Sp(\gamma)}(\I(\gamma)^{(k)})=n}}\hspace{-1cm}\I(\Sp(\gamma))^{(k)} , \;\;\;
	\I(\gamma) = \bigcup_{n \in \mathbb{Z}_q} \I_n(\gamma), \;\;\; \I'(\gamma) = \bigcup_{\substack{n \in \mathbb{Z}_q\\ n \neq q}} \I_n(\gamma). \;\;\;
	\end{equation}

\begin{definition}[\textit{External Contours}]
    A contour $\ctr$ is \textit{external} with respect to a family $\ctrb$ if $\Sp(\ctr) \cap V(\ctr') = \emptyset$ for every $\ctr' \in \ctrb\backslash \{\ctr\}$. We will denote $\Gamma^e$ the family of all external contours from a given family of contours $\Gamma$. 
    
\end{definition}

In the usual Peierls' argument, the spin-flip symmetry is exploited in order to extract the contribution of a contour to the energy of a configuration. We will do the same here, but the spin-flip will be replaced by a transformation in the configuration space. Given some $\sigma \in \Omega^q_\Lambda$ and a contour $\gamma \in \Gamma^e(\sigma)$, we define

\begin{equation*}
    \tau_{\gamma}(\sigma)_x : =
    \begin{cases}
        q & \text{ if } x \in \Sp(\gamma),\\
        \sigma_x - n & \text{ if } x \in \I_n(\gamma),\\
        \sigma_x & \text{ if } x \in V(\gamma)^c .
    \end{cases}
\end{equation*}

\begin{figure}[H]
    \centering
    \tikzset{every picture/.style={line width=0.75pt}} 

\begin{tikzpicture}[x=0.75pt,y=0.75pt,yscale=-1,xscale=1]

\draw  [color={rgb, 255:red, 74; green, 144; blue, 226 }  ,draw opacity=1 ][fill={rgb, 255:red, 199; green, 192; blue, 192 }  ,fill opacity=1 ][line width=1.5]  (152.49,82.16) .. controls (183.07,69.05) and (217.65,70.64) .. (247.78,77.82) .. controls (277.9,85) and (306.54,114.05) .. (315.53,138.19) .. controls (324.52,162.33) and (298.79,222.3) .. (255.13,246.21) .. controls (211.48,270.11) and (48.45,214.09) .. (42.02,186.34) .. controls (35.58,158.6) and (121.92,95.28) .. (152.49,82.16) -- cycle ;
\draw  [color={rgb, 255:red, 74; green, 144; blue, 226 }  ,draw opacity=1 ][fill={rgb, 255:red, 74; green, 144; blue, 226 }  ,fill opacity=0.46 ][line width=0.75]  (92.54,139.95) .. controls (142.2,97.42) and (144.66,148.79) .. (144.58,160.51) .. controls (144.49,172.24) and (146.81,195.84) .. (134.31,206.86) .. controls (121.81,217.88) and (105,217.11) .. (84.91,199.94) .. controls (64.82,182.76) and (42.87,182.48) .. (92.54,139.95) -- cycle ;
\draw  [color={rgb, 255:red, 208; green, 2; blue, 27 }  ,draw opacity=1 ][fill={rgb, 255:red, 208; green, 2; blue, 27 }  ,fill opacity=0.31 ][line width=0.75]  (168.77,99.52) .. controls (182.05,96.36) and (222.8,108.27) .. (225.19,134.41) .. controls (227.57,160.55) and (163.72,170.6) .. (158.43,157.59) .. controls (153.14,144.58) and (155.5,102.67) .. (168.77,99.52) -- cycle ;
\draw  [color={rgb, 255:red, 248; green, 231; blue, 28 }  ,draw opacity=1 ][fill={rgb, 255:red, 248; green, 231; blue, 28 }  ,fill opacity=0.51 ][line width=0.75]  (229.74,170.65) .. controls (242.11,177.22) and (265.17,220.77) .. (241.51,228.22) .. controls (217.85,235.67) and (166.36,243.88) .. (159.92,216.14) .. controls (153.49,188.4) and (217.36,164.08) .. (229.74,170.65) -- cycle ;
\draw  [color={rgb, 255:red, 74; green, 144; blue, 226 }  ,draw opacity=1 ][fill={rgb, 255:red, 74; green, 144; blue, 226 }  ,fill opacity=0.48 ][line width=0.75]  (244.05,101.43) .. controls (256.11,93.37) and (273.6,111.09) .. (300.73,132.84) .. controls (327.87,154.59) and (295.31,156.2) .. (287.47,191.64) .. controls (279.63,227.08) and (266.97,193.25) .. (254.91,169.09) .. controls (242.85,144.92) and (231.99,109.48) .. (244.05,101.43) -- cycle ;
\draw  [color={rgb, 255:red, 74; green, 144; blue, 226 }  ,draw opacity=1 ][fill={rgb, 255:red, 74; green, 144; blue, 226 }  ,fill opacity=1 ] (19,81.34) -- (35.28,81.34) -- (35.28,97.45) -- (19,97.45) -- cycle ;
\draw  [color={rgb, 255:red, 208; green, 2; blue, 27 }  ,draw opacity=1 ][fill={rgb, 255:red, 208; green, 2; blue, 27 }  ,fill opacity=1 ] (19,14.49) -- (35.28,14.49) -- (35.28,30.6) -- (19,30.6) -- cycle ;
\draw  [color={rgb, 255:red, 248; green, 231; blue, 28 }  ,draw opacity=1 ][fill={rgb, 255:red, 248; green, 231; blue, 28 }  ,fill opacity=1 ] (19,47.51) -- (35.28,47.51) -- (35.28,63.62) -- (19,63.62) -- cycle ;
\draw  [color={rgb, 255:red, 74; green, 144; blue, 226 }  ,draw opacity=1 ][fill={rgb, 255:red, 155; green, 155; blue, 155 }  ,fill opacity=1 ][line width=0.75]  (102.95,153.78) .. controls (110.19,145.73) and (126.47,173.92) .. (122.25,181.17) .. controls (118.03,188.42) and (112,193.25) .. (107.78,185.2) .. controls (103.56,177.14) and (95.72,161.84) .. (102.95,153.78) -- cycle ;
\draw  [color={rgb, 255:red, 74; green, 144; blue, 226 }  ,draw opacity=1 ][fill={rgb, 255:red, 155; green, 155; blue, 155 }  ,fill opacity=1 ][line width=0.75]  (263.35,130.42) .. controls (270.58,122.37) and (286.86,150.56) .. (282.64,157.81) .. controls (278.42,165.06) and (272.39,169.89) .. (268.17,161.84) .. controls (263.95,153.78) and (256.11,138.48) .. (263.35,130.42) -- cycle ;
\draw  [color={rgb, 255:red, 208; green, 2; blue, 27 }  ,draw opacity=1 ][fill={rgb, 255:red, 155; green, 155; blue, 155 }  ,fill opacity=1 ][line width=0.75]  (187.48,118.81) .. controls (200.74,118.81) and (203.88,124.93) .. (202.07,132.98) .. controls (200.26,141.04) and (167.85,146.81) .. (167.25,141.17) .. controls (166.64,135.53) and (174.21,118.81) .. (187.48,118.81) -- cycle ;
\draw  [color={rgb, 255:red, 248; green, 231; blue, 28 }  ,draw opacity=1 ][fill={rgb, 255:red, 248; green, 231; blue, 28 }  ,fill opacity=1 ] (183.89,126.68) .. controls (185.43,125.11) and (191.83,128.89) .. (190.81,131.09) .. controls (189.78,133.3) and (180.56,136.45) .. (178.77,133.3) .. controls (176.98,130.15) and (182.36,128.26) .. (183.89,126.68) -- cycle ;
\draw  [color={rgb, 255:red, 248; green, 231; blue, 28 }  ,draw opacity=1 ][fill={rgb, 255:red, 155; green, 155; blue, 155 }  ,fill opacity=1 ][line width=0.75]  (208.22,191.55) .. controls (221.49,191.55) and (224.63,197.66) .. (222.82,205.72) .. controls (221.01,213.77) and (188.59,219.54) .. (187.99,213.9) .. controls (187.39,208.26) and (194.95,191.55) .. (208.22,191.55) -- cycle ;
\draw  [color={rgb, 255:red, 248; green, 231; blue, 28 }  ,draw opacity=1 ][fill={rgb, 255:red, 248; green, 231; blue, 28 }  ,fill opacity=1 ] (204.63,200.36) .. controls (206.17,198.79) and (212.57,202.57) .. (211.55,204.77) .. controls (210.52,206.97) and (201.31,210.12) .. (199.51,206.97) .. controls (197.72,203.83) and (203.1,201.94) .. (204.63,200.36) -- cycle ;
\draw  [color={rgb, 255:red, 74; green, 144; blue, 226 }  ,draw opacity=1 ][fill={rgb, 255:red, 74; green, 144; blue, 226 }  ,fill opacity=0.8 ][line width=1.5]  (470.98,69.06) .. controls (501.56,55.94) and (536.14,57.53) .. (566.27,64.71) .. controls (596.39,71.9) and (625.03,100.95) .. (634.02,125.08) .. controls (643.01,149.22) and (617.28,209.19) .. (573.62,233.1) .. controls (529.97,257.01) and (366.94,200.98) .. (360.51,173.24) .. controls (354.07,145.49) and (440.41,82.17) .. (470.98,69.06) -- cycle ;
\draw  [color={rgb, 255:red, 74; green, 144; blue, 226 }  ,draw opacity=1 ][fill={rgb, 255:red, 155; green, 155; blue, 155 }  ,fill opacity=1 ][line width=0.75]  (420.64,140.68) .. controls (427.88,132.62) and (444.16,160.81) .. (439.94,168.06) .. controls (435.72,175.31) and (429.69,180.14) .. (425.47,172.09) .. controls (421.25,164.03) and (413.41,148.73) .. (420.64,140.68) -- cycle ;
\draw  [color={rgb, 255:red, 74; green, 144; blue, 226 }  ,draw opacity=1 ][fill={rgb, 255:red, 155; green, 155; blue, 155 }  ,fill opacity=1 ][line width=0.75]  (581.04,117.32) .. controls (588.27,109.26) and (604.55,137.45) .. (600.33,144.7) .. controls (596.11,151.95) and (590.08,156.79) .. (585.86,148.73) .. controls (581.64,140.68) and (573.8,125.37) .. (581.04,117.32) -- cycle ;
\draw  [color={rgb, 255:red, 74; green, 144; blue, 226 }  ,draw opacity=1 ][fill={rgb, 255:red, 155; green, 155; blue, 155 }  ,fill opacity=1 ][line width=0.75]  (505.17,105.7) .. controls (518.43,105.7) and (521.57,111.82) .. (519.76,119.87) .. controls (517.95,127.93) and (485.54,133.7) .. (484.94,128.06) .. controls (484.33,122.42) and (491.9,105.7) .. (505.17,105.7) -- cycle ;
\draw  [color={rgb, 255:red, 208; green, 2; blue, 27 }  ,draw opacity=1 ][fill={rgb, 255:red, 208; green, 2; blue, 27 }  ,fill opacity=1 ] (501.58,113.58) .. controls (503.12,112) and (509.52,115.78) .. (508.5,117.98) .. controls (507.47,120.19) and (498.25,123.34) .. (496.46,120.19) .. controls (494.67,117.04) and (500.05,115.15) .. (501.58,113.58) -- cycle ;
\draw  [color={rgb, 255:red, 74; green, 144; blue, 226 }  ,draw opacity=1 ][fill={rgb, 255:red, 155; green, 155; blue, 155 }  ,fill opacity=1 ][line width=0.75]  (525.91,178.44) .. controls (539.17,178.44) and (542.31,184.55) .. (540.51,192.61) .. controls (538.7,200.66) and (506.28,206.43) .. (505.68,200.79) .. controls (505.08,195.16) and (512.64,178.44) .. (525.91,178.44) -- cycle ;
\draw  [color={rgb, 255:red, 74; green, 144; blue, 226 }  ,draw opacity=1 ][fill={rgb, 255:red, 74; green, 144; blue, 226 }  ,fill opacity=1 ] (522.32,187.25) .. controls (523.86,185.68) and (530.26,189.46) .. (529.24,191.66) .. controls (528.21,193.87) and (519,197.02) .. (517.2,193.87) .. controls (515.41,190.72) and (520.79,188.83) .. (522.32,187.25) -- cycle ;
\draw [line width=1.5]    (339.8,17.2) -- (339.8,322.2) ;

\draw (42.38,15.64) node [anchor=north west][inner sep=0.75pt]  [font=\large]  {$1$};
\draw (42.98,48.94) node [anchor=north west][inner sep=0.75pt]  [font=\large]  {$2$};
\draw (42.98,82.91) node [anchor=north west][inner sep=0.75pt]  [font=\large]  {$3$};
\draw (183,293.4) node [anchor=north west][inner sep=0.75pt]  [font=\large]  {$\sigma $};
\draw (210,86.4) node [anchor=north west][inner sep=0.75pt]  [font=\large]  {$\gamma $};
\draw (473,288.4) node [anchor=north west][inner sep=0.75pt]  [font=\large]  {$\tau _{\gamma }( \sigma )$};

\end{tikzpicture}
    \caption{Notice that the effect of $\tau_{\gamma}$ in $\sigma$ is to erase the contour $\gamma$. In this example, a spin with color yellow inside a red interior becomes red, while it becomes blue being in a yellow interior. }
\end{figure}

The last feature of the contours we will need (and a very crucial one) is the exponential growth of the numbers of contours with a given size. Define

\[ 
\mathcal{C}_y(n) \coloneqq \{\gamma: y \in V(\gamma), |\gamma|=n\}.
\]
\begin{proposition}\label{Bound_on_C_0_n}
	Let $d\ge 2$, $y\in \Z^d$ and $\Lambda\Subset \mathbb{Z}^d$. There exists $c_1\coloneqq c_1(d,M,\alpha)>0$ such that
	\begin{equation}\label{Eq: exp.bound.contours}
	|\mathcal{C}_y(n)| \leq e^{(\log q + c_1) n}, \quad \forall\, n\geq 1.
	\end{equation}
\end{proposition}
\begin{proof}
Consider the projection $\Sp: \mathcal{C}_0(n) \to \mathcal{P}_f(\mathbb{Z}^d)$ given by $(\overline{\gamma}, \sigma_{\overline{\gamma}}) \mapsto \overline{\gamma}$. Then,

\begin{equation*}
    \mathcal{C}_0(n) = \bigcupdot_{A \in \Sp(\mathcal{C}_0(n))} \Sp^{-1}(A).
\end{equation*}

Therefore, 

\begin{equation*}
  |\mathcal{C}_0(n)| = \sum_{A \in \Sp(\mathcal{C}_0(n))} |\Sp^{-1}(A)|.
\end{equation*}

Now, note that $|\Sp^{-1}(A)| \leq |\Omega_A| = q^{|A|}=e^{|A|\log{q}}$. Hence, 
\[
|\mathcal{C}_0(n)| \leq \sum_{A \in \Sp(\mathcal{C}_0(n))}e^{|A|\log{q}}  = e^{n\log{q}} | \Sp(\mathcal{C}_0(n))|.
\]

Using the Corollary $3.29$ from \cite{Johanes}, we know that $| \Sp(\mathcal{C}_0(n))|\leq e^{c_1n}$.
Therefore,
\[
|\mathcal{C}_0(n)| = |\mathcal{C}_y(n)| \leq e^{n\log{q}}. e^{c_1n}=e^{(c_1+\log{q})n}.
\]  
\end{proof}

\section{Energy Bounds}\label{energy}

In this section we are going to prove the main bounds of this work. Before that, we present two useful lemmas.  Without loss of generality, we may suppose, by the addition of a constant and a suitable redefinition of $J$, that we can rewrite the Hamiltonian as

\begin{equation*}
    H_{\Lambda, h}^\eta(\sigma) = \sum_{\substack{\{x, y\} \subset \Lambda}}J_{xy}\psi(\sigma_x - \sigma_y) + \sum_{\substack{x \in \Lambda \\ y \notin \Lambda}} J_{xy} \psi(\sigma_x - \eta_y) + \sum_{x \in \Lambda} h_x \psi(\sigma_x),
\end{equation*}
with $\psi$ such that $ 0 \leq \psi(n) \leq 1$, for any $n \in \Z_q$ and $\psi(0) = 0$. Explicitly, we can take

\begin{equation*}
    \psi(n) = \frac{\varphi(0) - \varphi(n)}{\varphi(0) - \min_{n \neq 0} \varphi(n)}.
\end{equation*}

After this redefinition, we denote by $m := \min\{\psi(n); n \neq 0\}$ the minimum excitation. Observe that $m > 0$.

\begin{lemma}\label{lema_geometrico}
For any $x, y \in \mathbb{Z}^d$ such that $x \neq y$, it holds
    \begin{equation*}
    J_{xy} \geq \frac{1}{(2d+1)2^\alpha} \sum_{x'\in B_1(x)} J_{x'y}.
\end{equation*}
\end{lemma}
\begin{proof}
Firstly, notice that we have

\begin{equation*}
    \sum_{x'\in B_1(x)} J_{x'y} = J_{xy}  \sum_{x'\in B_1(x)} \frac{J_{x'y}}{J_{xy}} = J_{xy}  \sum_{x'\in B_1(x)\backslash \{y\}} \left( \frac{|x - y|}{|x'- y|} \right)^\alpha.
\end{equation*}

Using the triangle inequality, it follows that:

\begin{equation*}
    \sum_{x'\in B_1(x)} J_{x'y} \leq J_{xy}\sum_{x'\in B_1(x)\backslash \{y\}} \left( \frac{|x - x'|}{|x'- y|} + \frac{|x' - y|}{|x'- y|} \right)^\alpha = J_{xy}\sum_{x'\in B_1(x)\backslash \{y\}} \left( \frac{1}{|x'- y|} + 1 \right)^\alpha.
\end{equation*}

Since $1/|x'- y| \leq 1$, 

\begin{equation*}
    \sum_{x'\in B_1(x)} J_{x'y} \leq J_{xy} \sum_{x'\in B_1(x)\backslash \{y\}} 2^\alpha \leq J_{xy}(2d+1)2^\alpha,
\end{equation*}
and the inequality is proven.
\end{proof}

\begin{lemma}\label{lema_pontos_incorretos}
For any contour $\gamma$, and $y \in \mathbb{Z}^d$, it holds that
\begin{equation*}
    \sum_{\substack{x \in \Sp(\gamma)\\x'\in B_1(x)\\x \neq y}} J_{x'y}\psi(\sigma_x - \sigma_y) \geq m\sum_{z \in \Sp(\gamma)} J_{zy}.
\end{equation*}
\end{lemma}
\begin{proof}
Let $z$ be an element of $\Sp(\gamma)$. There exists $x = x(z)$ in $\Sp(\gamma) \cap B_1(z)$ such that $\sigma_x \neq \sigma_y$. In fact, if $\sigma_z \neq \sigma_y$, we can simply take $x(z) = z$. If $\sigma_z = \sigma_y$, since $z$ is an incorrect point, there exists $x \in B_1(z)$ such that $\sigma_x \neq \sigma_z$, so $\sigma_x \neq \sigma_y$. But $x$ will also be an incorrect point, so $x$ must be in $\Sp(\gamma)$. Remembering that $m = \min\{\psi(z); z \neq 0\}$, we conclude that, for each $z \in \Sp(\gamma)$, there exists $x(z)$ such that $J_{zy} \psi(\sigma_{x(z)} - \sigma_y) \geq m J_{zy}$. Summing over $z$,

\begin{equation*}
   m\hspace{-0.2cm} \sum_{z \in \Sp(\gamma)} J_{zy} \leq \sum_{z \in \Sp(\gamma)} J_{zy} \psi(\sigma_{x(z)} - \sigma_y).
\end{equation*}

Now, since every term is non-negative, we can get an upper bound by summing also over $x$, then

\begin{equation*}
  m \hspace{-0.2cm} \sum_{z \in \Sp(\gamma)} J_{zy} \leq  \sum_{\substack{z \in \Sp(\gamma)\\ x \in \Sp(\gamma) \\ |x - z| = 1}} J_{zy} \psi(\sigma_{x} - \sigma_y). 
\end{equation*}

Notice that the sum in the right-hand side is a sum over all ordered pairs $(x, z)$ such that $x, z \in \Sp(\gamma)$ and $|x - z| = 1$, but this is the same as summing over $x \in \Sp(\gamma)$ and then over $z \in B_1(x) \cap \Sp(\gamma)$. Again using that each term is non-negative, we can drop the last conditions and we have

\begin{equation*}
  m \hspace{-0.2cm}\sum_{z \in \Sp(\gamma)} J_{zy} \leq  \sum_{\substack{x \in \Sp(\gamma)\\ z \in B_1(x)}} J_{zy} \psi(\sigma_{x} - \sigma_y). 
\end{equation*}

\end{proof}

The following proposition, which gives us the energy of erasing a contour, will be the core of the Peierls' argument in the next section. 

\begin{proposition}\label{BE}
    There is a constant $c_2 = c_2(J, m, \alpha, d)$ such that, for any configuration $\sigma \in \Omega^q_{\Lambda}$ and $\gamma \in \Gamma^e(\sigma)$,

\begin{equation*}
    H^q_{\Lambda}(\sigma) - H^q_{\Lambda}(\tau) \geq c_2\left(|\gamma| +  F_{\Sp(\gamma)} + \sum_{n = 1}^{q-1}F_{\I_n(\gamma)} + F_{\I'(\gamma)}\right),
\end{equation*}
 where $\tau := \tau_{\gamma}(\sigma)$ and $\I'(\gamma)$ is given by Equation \eqref{interior}.
\end{proposition}

\begin{proof}

In first place, let's investigate how to write the Hamiltonian in terms of the contours. To do this, we will write the Hamiltonian in terms of the function $\psi$. Given subsets $A, B \Subset \mathbb{Z}^d$ and some configuration $\sigma \in \Omega_\Lambda^q$, we define

\begin{equation*}
    \psi(A, B)[\sigma] = \sum_{\substack{x \in A\\y \in B}} J_{xy}\psi(\sigma_x - \sigma_y)
\end{equation*}
and
\begin{equation*}
    \psi(A)[\sigma] = \frac{1}{2}\psi(A, A)[\sigma] = \sum_{\{x, y\} \subset A} J_{xy} \psi(\sigma_x - \sigma_y).
\end{equation*}

Then, for any partition $\Lambda = \bigcup_{k=1}^n \Lambda_k$ of $\Lambda$, the Hamiltonian decomposes as

\begin{equation*}
    H^q_{\Lambda}(\sigma) = \sum_{k = 1}^n \psi(\Lambda_k)[\sigma] + \sum_{\{i, j\}} \psi(\Lambda_i, \Lambda_j)[\sigma] + \sum_{k = 1}^n \psi(\Lambda_k, \Lambda^c)[\sigma].
\end{equation*}

We are interested in finding a lower bound for $H^q_{\Lambda}(\sigma) - H^q_{\Lambda}(\tau)$ depending only on $\gamma$. In order to do so, we are going to start by partitioning $\Lambda$ into $\Sp(\gamma) \cup \bigcup_{n = 1}^q \I_n(\gamma) \cup \Lambda\backslash V(\gamma)$. 

The previous remark gives us

\begin{align*}
    H^q_{\Lambda}(\sigma) &= \psi(\Sp(\gamma))[\sigma] + \sum_{n=1}^q \psi(\I_n(\gamma))[\sigma] + \psi(\Lambda\backslash V(\gamma))[\sigma] + \sum_{n=1}^q \psi(\Sp(\gamma), \I_n(\gamma))[\sigma] \\
    &+ \psi(\Sp(\gamma), \Lambda \backslash V(\gamma))[\sigma] + \sum_{n=1}^q \psi(\I_n(\gamma), \Lambda\backslash V(\gamma))[\sigma] + \sum_{n \neq n'} \psi(\I_n(\gamma), \I_{n'}(\gamma))[\sigma] \\
    &+ \psi(\Sp(\gamma), \Lambda^c)[\sigma] + \sum_{n = 1}^q \psi(\I_n(\gamma), \Lambda^c)[\sigma] + \psi(\Lambda\backslash V(\gamma), \Lambda^c)[\sigma],
\end{align*}

where $n \neq n'$ indicates a summation over unordered pairs $\{n, n'\}$ of distinct elements of $\{1, ..., q\}$.

Now, since we are interested in the difference of the Hamiltonians, let's define $\Delta(A, B)$ as $\psi(A, B)[\sigma] - \psi(A, B)[\tau]$ and $\Delta(A) = \Delta(A, A)/2$. Since the $\tau$ map leaves $\Lambda^c, \Lambda\backslash V(\gamma)$ and $\I_q(\gamma)$ invariant we know that any term which only depends on these regions will be cancelled out. In a less obvious fashion, notice that the $\tau$ map acts on each $\I_n$ as a translation and, since $\psi$ only depends on the difference between spins, $\psi(\sigma_x - \sigma_y) = \psi(\tau_x - \tau_y)$ whenever $x, y \in \I_n(\gamma)$. Thus, $\Delta(\I_n(\gamma)) =0 $. We are then left with

\begin{align*}
    H^q_{\Lambda}(\sigma) - H^q_\Lambda(\tau) &= \psi(\Sp(\gamma))[\sigma] + \sum_{n=1}^q \psi(\Sp(\gamma), \I_n(\gamma))[\sigma] + \psi(\Sp(\gamma), V(\gamma)^c)[\sigma]\\
    &- \sum_{n=1}^q \psi(\Sp(\gamma), \I_n(\gamma))[\tau] - \psi(\Sp(\gamma), V(\gamma)^c)[\tau]\\
    &+ \sum_{n=1}^{q-1} \Delta(\I_n(\gamma), \Lambda\backslash V(\gamma)) + \sum_{n=1}^{q-1} \Delta(\I_n(\gamma), \Lambda^c) + \sum_{n=1}^{q-1} \Delta(\I_n(\gamma), \I_q(\gamma))\\
    &+ \sum_{\{n, n'\} \subset \{1, ..., q-1\}} \Delta(\I_n(\gamma), \I_{n'}(\gamma)).
\end{align*}

We can consider the union $Q(\gamma) = \I_q(\gamma) \cup V(\gamma)^c$ and we rewrite the difference as

\begin{equation*}
    H^q_{\Lambda}(\sigma) - H^q_\Lambda(\tau) = \text{(I)} + \text{(II)} + \text{(III)},
\end{equation*}
where
\begin{align*}
   \text{(I)} &= \psi(\Sp(\gamma))[\sigma] + \sum_{n=1}^q \psi(\Sp(\gamma), \I_n(\gamma))[\sigma] + \psi(\Sp(\gamma), V(\gamma)^c)[\sigma]\\
    \text{(II)} &= - \sum_{n=1}^q \psi(\Sp(\gamma), \I_n(\gamma))[\tau] - \psi(\Sp(\gamma), V(\gamma)^c)[\tau]\\
    \text{(III)} &=  \sum_{n=1}^{q-1} \Delta(\I_n(\gamma), Q(\gamma)) + \sum_{\{n, n'\} \subset \{1, ..., q-1\}} \Delta(\I_n(\gamma), \I_{n'}(\gamma)).
\end{align*}

Now, we will bound \text{(I)}, \text{(II)} e \text{(III)}. The first line is equal to

\begin{equation}\label{(I)}
    \frac{1}{2} \sum_{\substack{x \in \Sp(\gamma)\\y \in \mathbb{Z}^d}} J_{xy}\psi(\sigma_x - \sigma_y) + \frac{1}{2} \sum_{\substack{x \in \Sp(\gamma)\\y \in \Sp(\gamma)^c}} J_{xy}\psi(\sigma_x - \sigma_y),
\end{equation}
so we face the task to provide a lower bound for the expression above. Clearly, many terms above will be zero --- always that we have a pair of equal spins. However, we can use the fact that the contour is composed of incorrect points to see that, given a pair $\{x, y\}$ of sites with the same spin and $x \in \Sp(\gamma)$, there exists a $x' \in B_1(x)$ such that $\{x', y\}$ is a pair of sites with different spins. Hence, it will be useful to consider averages of interactions across balls.

Now, using the previous inequality and the Lemma \ref{lema_pontos_incorretos}, we have  
\begin{align*}
\text{(I)} &= \frac{1}{2} \sum\limits_{y \in \mathbb{Z}^d} \sum\limits_{x \in \Sp(\gamma)} J_{xy} \psi(\sigma_x - \sigma_y)+ \frac{1}{2} \sum\limits_{y \in \Sp(\gamma)^c} \sum\limits_{x \in \Sp(\gamma)} J_{xy} \psi(\sigma_x - \sigma_y)\\[0.2cm]
&\geq \frac{1}{2} \sum\limits_{y \in \mathbb{Z}^d} \frac{1}{(2d+1)2^\alpha}\sum_{\substack{x \in \Sp(\gamma)\\x'\in B_1(x)\\x \neq y}} J_{x'y}\psi(\sigma_x - \sigma_y) + \frac{1}{2} \sum\limits_{y \in \Sp(\gamma)^c} \frac{1}{(2d+1)2^\alpha}\sum_{\substack{x \in \Sp(\gamma)\\x'\in B_1(x)\\x \neq y}} J_{x'y}\psi(\sigma_x - \sigma_y)\\[0.2cm]
&\geq \frac{1}{2} \sum\limits_{y \in \mathbb{Z}^d} \frac{m}{(2d+1)2^{\alpha}} \sum\limits_{z \in \Sp(\gamma)} J_{zy} + \frac{1}{2} \sum\limits_{y \in \Sp(\gamma)^c} \frac{m}{(2d+1)2^{\alpha}} \sum\limits_{z \in \Sp(\gamma)} J_{zy}.
\end{align*}
Thus
\begin{equation*}
\text{(I)} \geq \frac{m}{(2d+1)2^{\alpha + 1}}\left(\sum_{z\in \Sp(\gamma)}\sum_{y \in \mathbb{Z}^d} J_{zy} + \sum_{\substack{z \in \Sp(\gamma)\\ y \in \Sp(\gamma)^c}}J_{zy}\right)   \geq \frac{m}{(2d+1)2^{\alpha + 1}}\left(J c_\alpha |\gamma| + F_{\Sp(\gamma)}\right),
\end{equation*}
where $c_\alpha := \sum_{y \neq 0} |y|^{-\alpha}$.

For the second term, we have
\begin{align*}
\text{(II)} &= - \sum\limits_{n=1}^q \sum\limits_{\substack{x \in \Sp(\gamma) \\ y \in \I_n(\gamma)}} J_{xy} \psi(\tau_x - \tau_y)
- \sum\limits_{\substack{x \in \Sp(\gamma) \\ y \in V(\gamma)^c}} J_{xy} \psi(\tau_x - \tau_y)\\[0.2cm]
&= - \sum\limits_{\substack{x \in \Sp(\gamma) \\ y \in I(\gamma)}} J_{xy} \psi(\tau_x - \tau_y)
- \sum\limits_{\substack{x \in \Sp(\gamma) \\ y \in V(\gamma)^c}} J_{xy} \psi(\tau_x - \tau_y)\\[0.2cm]
&= - \sum_{\substack{x \in \Sp(\gamma)\\y \in \Sp(\gamma)^c}} J_{xy} \psi(\tau_x - \tau_y) = - \sum_{\substack{x \in \Sp(\gamma)\\y \in \Sp(\gamma)^c}} J_{xy} \psi(q - \tau_y),
\end{align*}
where we used the definition of the $\tau$ map.

Now, putting $\Gamma = \Gamma(\sigma)$, it's not difficult to see that $\tau_y \neq q$ implies that $y \in V(\Gamma \backslash \gamma)$. Thus, the summation is zero for any $y \notin V(\Gamma \backslash \gamma )$. This observation, together with Corollary 2.9 of \cite{Johanes}, gives us
\begin{equation}
  \sum\limits_{\substack{x \in \Sp(\gamma) \\ y \in \Sp(\gamma)^c}} J_{xy} \psi(\tau_x - \tau_y) \leq 
\sum_{\substack{x \in \Sp(\gamma)\\y \in V(\Gamma \backslash \gamma )}} J_{xy} \leq \kappa^{(2)}_\alpha \frac{F_{\Sp(\gamma)}}{M^{(\alpha - d) \land 1}}, 
\end{equation}
where 
\begin{equation*}
\kappa^{(2)}_\alpha := (1 + J^{-1})\left[\frac{J2^{d-1+\alpha}e^{d-1}}{(\alpha - d)} + 3\zeta\left(\frac{a}{d+1} - 1\right)\right].
\end{equation*}
Hence,
\begin{align*}
\text{(II)} &= - \sum\limits_{\substack{x \in sp(\gamma) \\ y \in sp(\gamma)^c}} J_{xy} \psi(\tau_x - \tau_y) \geq - \sum\limits_{\substack{x \in sp(\gamma) \\ y \in V(\Gamma\backslash\gamma)}} J_{xy} \geq - \kappa_{\alpha}^{(2)} \frac{F_{\Sp(\gamma)}}{M^{(\alpha - d)\land 1} }\\[0.2cm]
\end{align*}

As for the third term,
\begin{align*}
    \text{(III)} &= \frac{1}{2}\sum_{n = 1}^{q-1}\sum_{\substack{n'= 1\\n'\neq n}}^{q-1} \Delta(\I_n(\gamma), \I_{n'}(\gamma)) + \sum_{n =1}^{q-1} \Delta(\I_n(\gamma), Q(\gamma))\\
    &= \frac{1}{2}\sum_{n = 1}^{q-1} \Delta(\I_n(\gamma), (I_n(\gamma) \cup \Sp(\gamma))^c) + \frac{1}{2}\sum_{n =1}^{q-1} \Delta(\I_n(\gamma), Q(\gamma))\\[0.3cm]
    &= \frac{1}{2} \sum_{n = 1}^{q-1} A_n(\gamma) + \frac{1}{2} B(\gamma),
\end{align*}

where

\begin{equation*}
   A_n(\gamma) = \sum_{\substack{x \in \I_n(\gamma)\\y \notin I_n(\gamma) \cup \Sp(\gamma)}} J_{xy} \left(\psi(\sigma_x - \sigma_y) - \psi(\tau_x - \tau_y)\right), 
\end{equation*}
and
\begin{equation*}
    B(\gamma) = \sum_{\substack{x \in \I'(\gamma)\\y \in Q(\gamma)}} J_{xy} \left(\psi(\sigma_x - \sigma_y) - \psi(\tau_x - \tau_y)\right).
\end{equation*}

Fixed some $n$, in order to bound $A_n(\gamma)$ we use $\Gamma'$ to denote the set of contours inside $\I_n(\gamma)$ and $\Gamma''$ to denote the set of contours outside $\I_n(\gamma)$ (except for $\gamma$). Outside of the volumes of $\Gamma'$ and $\Gamma''$, the spins are controllable, that is,
\begin{equation*}
\sigma_y = 
\begin{dcases}
n &\text{ if } y \in \I_n(\gamma) \backslash V(\Gamma'), \\
n' &\text{ if } y \in \I_{n'}(\gamma)\backslash V(\Gamma''),\\
q &\text{ if } y \in Q(\gamma)\backslash V(\Gamma'').
\end{dcases}
\end{equation*}

\begin{figure}
    \centering
 
\tikzset{
pattern size/.store in=\mcSize, 
pattern size = 5pt,
pattern thickness/.store in=\mcThickness, 
pattern thickness = 0.3pt,
pattern radius/.store in=\mcRadius, 
pattern radius = 1pt}
\makeatletter
\pgfutil@ifundefined{pgf@pattern@name@_kbgqzonu9}{
\pgfdeclarepatternformonly[\mcThickness,\mcSize]{_kbgqzonu9}
{\pgfqpoint{0pt}{0pt}}
{\pgfpoint{\mcSize}{\mcSize}}
{\pgfpoint{\mcSize}{\mcSize}}
{
\pgfsetcolor{\tikz@pattern@color}
\pgfsetlinewidth{\mcThickness}
\pgfpathmoveto{\pgfqpoint{0pt}{\mcSize}}
\pgfpathlineto{\pgfpoint{\mcSize+\mcThickness}{-\mcThickness}}
\pgfpathmoveto{\pgfqpoint{0pt}{0pt}}
\pgfpathlineto{\pgfpoint{\mcSize+\mcThickness}{\mcSize+\mcThickness}}
\pgfusepath{stroke}
}}
\makeatother

 
\tikzset{
pattern size/.store in=\mcSize, 
pattern size = 5pt,
pattern thickness/.store in=\mcThickness, 
pattern thickness = 0.3pt,
pattern radius/.store in=\mcRadius, 
pattern radius = 1pt}
\makeatletter
\pgfutil@ifundefined{pgf@pattern@name@_54f78mom5}{
\pgfdeclarepatternformonly[\mcThickness,\mcSize]{_54f78mom5}
{\pgfqpoint{0pt}{0pt}}
{\pgfpoint{\mcSize}{\mcSize}}
{\pgfpoint{\mcSize}{\mcSize}}
{
\pgfsetcolor{\tikz@pattern@color}
\pgfsetlinewidth{\mcThickness}
\pgfpathmoveto{\pgfqpoint{0pt}{\mcSize}}
\pgfpathlineto{\pgfpoint{\mcSize+\mcThickness}{-\mcThickness}}
\pgfpathmoveto{\pgfqpoint{0pt}{0pt}}
\pgfpathlineto{\pgfpoint{\mcSize+\mcThickness}{\mcSize+\mcThickness}}
\pgfusepath{stroke}
}}
\makeatother
\tikzset{every picture/.style={line width=0.75pt}} 

\begin{tikzpicture}[x=0.75pt,y=0.75pt,yscale=-1,xscale=1]

\draw  [fill={rgb, 255:red, 192; green, 192; blue, 192 }  ,fill opacity=1 ] (241.5,119.5) .. controls (261.5,109.5) and (275.5,90.5) .. (441.5,95.5) .. controls (607.5,100.5) and (547.5,123.5) .. (557.5,162.5) .. controls (567.5,201.5) and (492.5,233.5) .. (420.5,254.5) .. controls (348.5,275.5) and (185.5,282.5) .. (165.5,252.5) .. controls (145.5,222.5) and (221.5,129.5) .. (241.5,119.5) -- cycle ;
\draw  [fill={rgb, 255:red, 255; green, 255; blue, 255 }  ,fill opacity=1 ][dash pattern={on 4.5pt off 4.5pt}] (250,136) .. controls (270,126) and (290.66,117.43) .. (330,112) .. controls (332.47,111.66) and (335.14,111.34) .. (337.97,111.04) .. controls (380.28,106.53) and (459.52,106.39) .. (464,112) .. controls (468.78,117.98) and (433.72,146.69) .. (394,159) .. controls (354.28,171.31) and (304.48,177.55) .. (284,176) .. controls (263.52,174.45) and (230,146) .. (250,136) -- cycle ;
\draw  [fill={rgb, 255:red, 255; green, 255; blue, 255 }  ,fill opacity=1 ] (214.5,183.5) .. controls (238.5,135.5) and (242.5,225.5) .. (256.5,214.5) .. controls (270.5,203.5) and (334.7,185.92) .. (342.5,193.5) .. controls (350.3,201.08) and (322.5,234.5) .. (289.5,247.5) .. controls (256.5,260.5) and (237.58,251.19) .. (219.5,248.5) .. controls (201.42,245.81) and (190.5,231.5) .. (214.5,183.5) -- cycle ;
\draw  [fill={rgb, 255:red, 192; green, 192; blue, 192 }  ,fill opacity=1 ] (572.5,144.5) .. controls (598.5,129.5) and (637.5,171.5) .. (611.5,194.5) .. controls (585.5,217.5) and (462.5,277.5) .. (538.5,220.5) .. controls (614.5,163.5) and (546.5,159.5) .. (572.5,144.5) -- cycle ;
\draw  [fill={rgb, 255:red, 255; green, 255; blue, 255 }  ,fill opacity=1 ] (417.5,169.5) .. controls (429.95,163.27) and (527.5,149.5) .. (507.5,169.5) .. controls (487.5,189.5) and (478.53,221.48) .. (451,233.5) .. controls (423.47,245.52) and (415.86,243.58) .. (411,244.5) .. controls (406.14,245.42) and (395.17,233.15) .. (390.5,223.5) .. controls (385.83,213.85) and (405.05,175.73) .. (417.5,169.5) -- cycle ;
\draw  [pattern=_kbgqzonu9,pattern size=6pt,pattern thickness=0.75pt,pattern radius=0pt, pattern color={rgb, 255:red, 0; green, 0; blue, 0}] (310,132) .. controls (316,129) and (338.5,127) .. (328.5,141) .. controls (318.5,155) and (301,155) .. (296,146) .. controls (291,137) and (304,135) .. (310,132) -- cycle ;
\draw  [pattern=_54f78mom5,pattern size=6pt,pattern thickness=0.75pt,pattern radius=0pt, pattern color={rgb, 255:red, 0; green, 0; blue, 0}] (363,142) .. controls (360.5,136.5) and (387,120) .. (392,124.5) .. controls (397,129) and (378.5,149) .. (375,151) .. controls (371.5,153) and (365.5,147.5) .. (363,142) -- cycle ;
\draw  [color={rgb, 255:red, 0; green, 0; blue, 0 }  ,draw opacity=1 ][fill={rgb, 255:red, 0; green, 0; blue, 0 }  ,fill opacity=1 ] (426.5,209.5) .. controls (425,202.5) and (435.5,192) .. (447,196) .. controls (458.5,200) and (450.5,213.5) .. (444,215.5) .. controls (437.5,217.5) and (428,216.5) .. (426.5,209.5) -- cycle ;
\draw  [fill={rgb, 255:red, 0; green, 0; blue, 0 }  ,fill opacity=1 ] (279,233.5) .. controls (271,239.5) and (239,238) .. (244,228) .. controls (249,218) and (254,232.5) .. (264.5,231.5) .. controls (275,230.5) and (287,227.5) .. (279,233.5) -- cycle ;
\draw  [fill={rgb, 255:red, 0; green, 0; blue, 0 }  ,fill opacity=1 ] (54,33.5) .. controls (61.5,30.5) and (89,12.5) .. (105,15.5) .. controls (121,18.5) and (121,27.5) .. (103,42) .. controls (85,56.5) and (43,83.5) .. (38.5,73.5) .. controls (34,63.5) and (46.5,36.5) .. (54,33.5) -- cycle ;
\draw  [fill={rgb, 255:red, 0; green, 0; blue, 0 }  ,fill opacity=1 ] (124,30) .. controls (131.5,25) and (137,33) .. (136,40) .. controls (135,47) and (129,47.5) .. (124.5,40.5) .. controls (120,33.5) and (116.5,35) .. (124,30) -- cycle ;

\draw (567,103.9) node [anchor=north west][inner sep=0.75pt]  [font=\Large]  {$\gamma $};
\draw (337.5,132.4) node [anchor=north west][inner sep=0.75pt]    {$\Gamma '$};
\draw (453,180.4) node [anchor=north west][inner sep=0.75pt]  [font=\small]  {$\Gamma ''$};
\draw (287,219.4) node [anchor=north west][inner sep=0.75pt]  [font=\small]  {$\Gamma ''$};
\draw (99.5,54.9) node [anchor=north west][inner sep=0.75pt]  [font=\small]  {$\Gamma ''$};

\end{tikzpicture}
    \caption{The grouping of contours performed to bound $A_n(\gamma)$. The contour $\gamma$ is painted gray. The interior $\I_n(\gamma)$ is highlighted with a dashed line, and the family of contours inside it, denoted by $\Gamma'$, is filled with a checkered background. In contrast, the family of contours $\Gamma''$, outside $\I_n(\gamma)$ is filled with a solid black color.}
    \label{contour}
\end{figure}

This motivates us to split $A_n(\gamma)$ in terms of this sets. Explicitly,
\begin{align*}
A_{n}(\gamma) &= \sum\limits_{\substack{x \in \I_n(\gamma) \\ y \notin \I_n(\gamma) \cup \Sp(\gamma)}} \hspace{-0.5cm} J_{xy} \left( \psi(\sigma_x - \sigma_y) - \psi(\tau_x - \tau_y) \right)\\[0.2cm]
&= \sum_{\substack{x \in \I_n(\gamma)\\y \in V(\Gamma'')}} J_{xy} \left(\psi(\sigma_x - \sigma_y) - \psi(\sigma_x - \sigma_y)\right) \hspace{0.5cm} + \hspace{-0.7cm} \sum_{\substack{x \in \I_n(\gamma)\\y \notin \I_n(\gamma) \cup \Sp(\gamma) \cup V(\Gamma'')}} \hspace{-1cm}J_{xy} \left(\psi(\sigma_x - \sigma_y) - \psi(\sigma_x - \sigma_y)\right)\\[0.4cm]
&= \sum_{\substack{x \in \I_n(\gamma)\\y \in V(\Gamma'')}} J_{xy} \left(\psi(\sigma_x - \sigma_y) - \psi(\sigma_x - \sigma_y)\right) \hspace{0.5cm} + \hspace{-0.7cm}\sum_{\substack{x \in V(\Gamma')\\y \notin \I_n(\gamma) \cup \Sp(\gamma) \cup V(\Gamma'')}} \hspace{-1cm} J_{xy} \left(\psi(\sigma_x - \sigma_y) - \psi(\sigma_x - \sigma_y)\right)\\[0.5cm]
&+ \hspace{-0.7cm}\sum_{\substack{x \in \I_n(\gamma)\backslash V(\Gamma')\\y \notin \I_n(\gamma) \cup \Sp(\gamma) \cup V(\Gamma'')}}\hspace{-0.8cm} J_{xy} \left(\psi(\sigma_x - \sigma_y) - \psi(\sigma_x - \sigma_y)\right).
\end{align*}

In the first two summations we will use $\psi(\sigma_x - \sigma_y) - \psi(\tau_x - \tau_y) \geq -1$. In the last one, we know that $\psi(\sigma_x - \sigma_y) - \psi(\tau_x - \tau_y) = \psi(n - n') - 0 \geq m$. Thus,

\begin{equation}\label{an_parcial}
    A_n(\gamma) \hspace{0.3cm} \geq \hspace{-0.5cm}\sum_{\substack{x \in \I_n(\gamma)\backslash V(\Gamma')\\y \notin I_n(\gamma) \cup \Sp(\gamma) \cup V(\Gamma'')}} \hspace{-1cm}mJ_{xy} \hspace{0.3cm} -\sum_{\substack{x \in \I_n(\gamma)\\y \in V(\Gamma'')}}\hspace{-0.3cm} J_{xy} \hspace{0.3cm} - \hspace{-0.7cm}\sum_{\substack{x \in V(\Gamma')\\y \notin I_n(\gamma) \cup \Sp(\gamma) \cup V(\Gamma'')}}\hspace{-1cm} J_{xy}. 
\end{equation}

Now, notice that
\begin{align*}
F_{\I_n(\gamma)} = \sum\limits_{\substack{x \in \I_n(\gamma) \\ y \in \I_n(\gamma)^c}} J_{xy} \hspace{0.2cm} &= \hspace{-0.9cm} \sum\limits_{\substack{x \in \I_n(\gamma) \\ y \notin \I_n(\gamma) \cup \Sp(\gamma) \cup V(\Gamma'')}}\hspace{-1cm} J_{xy}\hspace{0.3cm}
+ \sum\limits_{\substack{x \in \I_n(\gamma) \\ y \in V(\Gamma'')}} \hspace{-0.2cm}J_{xy}\hspace{0.2cm}
+ \sum\limits_{\substack{x \in \I_n(\gamma) \\ y \in \Sp(\gamma)}}\hspace{-0.2cm} J_{xy}\\[0.4cm]
&= \hspace{-0.9cm}\sum\limits_{\substack{x \in \I_n(\gamma)\backslash V(\Gamma') \\ y \notin \I_n(\gamma) \cup \Sp(\gamma) \cup V(\Gamma'')}} \hspace{-1cm}J_{xy}\hspace{0.6cm} + \hspace{-0.5cm} \sum\limits_{\substack{x \in V(\Gamma') \\ y \notin \I_n(\gamma) \cup \Sp(\gamma) \cup V(\Gamma'')}} \hspace{-1cm}J_{xy}\hspace{0.3cm}
+ \sum\limits_{\substack{x \in \I_n(\gamma) \\ y \in V(\Gamma'')}} J_{xy}\hspace{0.2cm}
+ \sum\limits_{\substack{x \in \I_n(\gamma) \\ y \in \Sp(\gamma)}} J_{xy}.
\end{align*}

Rearranging, we are left with
\begin{equation*}
\sum_{\substack{x \in \I_n(\gamma)\backslash V(\Gamma')\\y \notin \I_n(\gamma) \cup \Sp(\gamma) \cup V(\Gamma'')}} \hspace{-1cm}J_{xy}\hspace{0.3cm} = \hspace{0.3cm} F_{\I_n(\gamma)} \hspace{0.2cm} - \hspace{-0.1cm} \sum_{\substack{x \in \I_n(\gamma)\\y \in  V(\Gamma'')}} \hspace{-0.2cm}J_{xy}\hspace{0.4cm} -  \hspace{-0.7cm}\sum_{\substack{x \in  V(\Gamma')\\y \notin \I_n(\gamma) \cup \Sp(\gamma) \cup V(\Gamma'')}}\hspace{-1cm} J_{xy} \hspace{0.3cm} - \sum_{\substack{x \in \I_n(\gamma)\\y \in \Sp(\gamma)}} J_{xy}\\[0.2cm]
\end{equation*}

Now, substituting the last expression in Equation \eqref{an_parcial},
\begin{align*}
A_n(\gamma) &\geq  m\left( F_{\I_n(\gamma)} - \sum\limits_{\substack{x \in \I_n(\gamma) \\ y \in \Sp(\gamma)}} J_{xy} \right) - (1+m)\left(
\sum\limits_{\substack{x \in \I_n(\gamma) \\ y \in V(\Gamma'')}} J_{xy} + \sum\limits_{\substack{x \in V(\Gamma') \\ y \notin \I_n(\gamma) \cup \Sp(\gamma) \cup V(\Gamma'')}} J_{xy}\right)\\
&\geq \frac{m}{(2d+1)2^{\alpha + 2}} \left( F_{\I_n(\gamma)} - \sum\limits_{\substack{x \in \I_n(\gamma) \\ y \in \Sp(\gamma)}} J_{xy} \right) - (1+m) \left(
\sum\limits_{\substack{x \in \I_n(\gamma) \\ y \in V(\Gamma'')}} J_{xy} + \sum\limits_{\substack{x \in V(\Gamma') \\ y \notin \I_n(\gamma) \cup \Sp(\gamma) \cup V(\Gamma'')}} J_{xy}\right)
\end{align*}

Again using Corollary 2.9 from \cite{Johanes},

\begin{align*}
    A_n(\gamma) &\geq \frac{m}{(2d+1)2^{\alpha + 2}}\left(F_{\I_n(\gamma)} - \sum_{\substack{x \in \I_n(\gamma)\\y \in \Sp(\gamma)}} J_{xy} \right) - (1+m) \left(\kappa^{(2)}_{\alpha} \frac{F_{\I_n(\gamma)}}{M^{(\alpha - d)\land 1}}  + \kappa^{(2)}_{\alpha} \frac{F_{\I_n(\gamma)}}{M} \right) \\
    &\geq \frac{m}{(2d+1)2^{\alpha + 2}}\left(F_{\I_n(\gamma)} - \sum_{\substack{x \in \I_n(\gamma)\\y \in \Sp(\gamma)}} J_{xy} \right) - 2(1+m) \kappa^{(2)}_{\alpha} \frac{F_{\I_n(\gamma)}}{M^{(\alpha - d)\land 1}}  \\
    &\geq \left( \frac{m}{(2d+1)2^{\alpha + 2}} - \frac{4\kappa^{(2)}_{\alpha}}{M^{(\alpha - d)\land 1}}\right)F_{\I_n(\gamma)} -  \frac{m}{(2d+1)2^{\alpha + 2}}\sum_{\substack{x \in \I_n(\gamma)\\y \in\Sp(\gamma)}} J_{xy}.
\end{align*}

Then,

\begin{align*}
    \sum\limits_{n=1}^{q-1} A_n(\gamma) &\geq \left(\frac{m}{(2d+1)2^{\alpha + 2}} - \frac{4\kappa^{(2)}_{\alpha}}{M^{(\alpha - d)\land 1}}\right) \sum\limits_{n=1}^{q-1} F_{\I_n(\gamma)} -  \frac{m}{(2d+1)2^{\alpha + 2}}\sum\limits_{n=1}^{q-1} \sum_{\substack{x \in \I_n(\gamma)\\y \in \Sp(\gamma)}} J_{xy}\\
    &= \left(\frac{m}{(2d+1)2^{\alpha + 2}} - \frac{4\kappa^{(2)}_{\alpha}}{M^{(\alpha - d)\land 1}}\right) \sum\limits_{m=1}^{q-1} F_{\I_n(\gamma)} -  \frac{m}{(2d+1)2^{\alpha + 2}} \sum_{\substack{x \in  \bigcupdot_{n = 1}^{q-1}\I_n(\gamma)\\y \in \Sp(\gamma)}} J_{xy}\\
    &\geq \left(\frac{m}{(2d+1)2^{\alpha + 2}} - \frac{4\kappa^{(2)}_{\alpha}}{M^{(\alpha - d)\land 1}}\right) \sum\limits_{m=1}^{q-1} F_{\I_n(\gamma)} -  \frac{m}{(2d+1)2^{\alpha + 2}} F_{\Sp(\gamma)}.
\end{align*}

The bound for $B(\gamma)$ is completely analogous, yielding

\begin{align*}
B(\gamma) \geq \left( \frac{m}{(2d+1)2^{\alpha + 2}} - \frac{4\kappa^{(2)}_{\alpha}}{M^{(\alpha - d)\land 1}}\right)F_{\I'(\gamma)} -  \frac{m}{(2d+1)2^{\alpha + 2}}F_{\Sp(\gamma)}.
\end{align*}
Finally, we are left with
\begin{align*}
\text{(III)} &= \frac{1}{2} \sum\limits_{m=1}^{q-1} A_n(\gamma) + \frac{1}{2} B(\gamma)\\[0.2cm]
& \geq \frac{1}{2} \left( \frac{m}{(2d+1)2^{\alpha+2}} - \frac{4\kappa_\alpha^{(2)}}{M^{(\alpha - d) \land 1}} \right) \sum\limits_{n=1}^{q-1} F_{\substack{\I_n(\gamma)}}
- \frac{m}{(2d+1)2^{\alpha+3}}
F_{\Sp(\gamma)} \\[0.2cm]
&+ \frac{1}{2} \left( \frac{m}{(2d+1)2^{\alpha+2}} - \frac{4\kappa_\alpha^{(2)}}{M^{(\alpha-d)\land 1}} \right) F_{\I'(\gamma)} - \frac{m}{(2d+1)2^{\alpha+3}} F_{\Sp(\gamma)}\\[0.2cm]
&\geq \left( \frac{m}{(2d+1)2^{\alpha + 3}} - \frac{2\kappa^{(2)}_{\alpha}}{M^{(\alpha - d)\land 1}}\right)\left(\sum_{n = 1}^{q-1}F_{\I_n(\gamma)} + F_{\I'(\gamma)}\right)-  \frac{m}{(2d+1)2^{\alpha + 2}}F_{\Sp(\gamma)}.
\end{align*}
Since
$$
H^q_{\Lambda}(\sigma) - H^q_{\Lambda}(\tau) = \text{(I)} + \text{(II)} + \text{(III)},
$$
we obtain that
\begin{align*}
H^q_{\Lambda}(\sigma) - H^q_{\Lambda}(\tau) &\geq \frac{m}{(2d+1)2^{\alpha + 1}}\left(J c_\alpha |\gamma| + F_{\Sp(\gamma)}\right) - \kappa_\alpha^{(2)} \frac{F_{\Sp(\gamma)}}{M^{(\alpha-d) \land 1}}\\[0.2cm]
&+ \left( \frac{m}{(2d+1)2^{\alpha+3}} - \frac{2 \kappa_{\alpha}^{(2)}}{M^{(\alpha-d) \Lambda 1}} \right) \left( \sum\limits_{n=1}^{q-1} F_{\I_n(\gamma)} + F_{\I'(\gamma)} \right) - \frac{m}{(2d+1)2^{\alpha+2}} F_{\Sp(\gamma)}\\[0.2cm]
&\geq \frac{Jmc_{\alpha}}{(2d + 1) 2^{\alpha + 1}} |\gamma| + \left( \frac{m}{(2d+1)2^{\alpha+1}} - \frac{m}{(2d+1)2^{\alpha+2}} - \frac{ \kappa_{\alpha}^{(2)}}{M^{(\alpha-d) \land 1}} \right) F_{\Sp(\gamma)}\\[0.2cm]
&+ \left( \frac{m}{(2d+1)2^{\alpha+3}} - \frac{2\kappa_{\alpha}^{(2)}}{M^{(\alpha-d) \land 1}} \right) \left( \sum\limits_{n=1}^{q-1} F_{\I_n(\gamma)} + F_{\I'(\gamma)} \right).
\end{align*}
Thus, we conclude that
\begin{align*}
H^q_{\Lambda}(\sigma) - H^q_{\Lambda}(\tau) &\geq \frac{Jmc_{\alpha}}{(2d+1)2^{\alpha + 1}}|\gamma| + \left(\frac{m}{(2d+1)2^{\alpha + 2}} - \frac{\kappa^{(2)}_\alpha}{M^{(\alpha - d)\land 1}} \right) F_{\Sp(\gamma)}\\[0.1cm]
&+ \left( \frac{m}{(2d+1)2^{\alpha + 3}} - \frac{2\kappa^{(2)}_{\alpha}}{M^{(\alpha - d)\land 1}}\right)\left(\sum_{n = 1}^{q-1}F_{\I_n(\gamma)} + F_{\I'(\gamma)}\right).
\end{align*}
Taking $M^{(\alpha-d)\wedge 1}> 2^{\alpha+5}(2d+1)\kappa_\alpha^{(2)}m^{-1}$ and $ c_2=\frac{m}{(2d+1)2^{\alpha+1}}\min\left\{Jc_\alpha, \frac{1}{8} \right\}$, the result of the demonstration follows.

\end{proof}

\begin{remark}
Although, for the phase transition result, the only relevant term in the upper bound is the one containing the support of the contour, we emphasize that this refined version could be significant when further details are required. For instance, the term $F_{\Sp(\gamma)}$ is crucial for obtaining the correct exponent for surface order large deviations in long-range ferromagnetic Ising spin systems \cite{large_deviations}. In the case of $q$-state spin systems, an additional term $F_{I'(\gamma)}$ appears, which is absent in the $q = 2$ case, and its potential impacts are yet unclear.
\end{remark}

\section{Phase Transition}\label{pt}
In this section we prove Theorem \ref{thm:phase_transition}, that is, the long-range Potts model with zero field undergoes a phase transition at low temperature. More precisely, we are going to prove that, for any $r, \ell \in \{1,\ldots,q\}$, if $r \neq \ell$, then the thermodynamic limits, $\mu^r_{\beta}$ and $\mu^\ell_{\beta}$, are also different when $\beta$ is large enough.

\begin{proof}[Proof of Theorem \ref{thm:phase_transition}]

The proof of Equation \eqref{eq_thm1} is the standard Peierls' argument. In fact, if $\sigma_{\Lambda^c}=r$ and $\sigma_0\neq r$, then there must exist a contour $\gamma$ with $0\in V(\gamma)$. Then,
$$
\mu_{\Lambda,\beta}^r(\sigma_0\neq r)\leq \mu^r_{\Lambda,\beta}\left(\left\{\sigma \in \Omega^r_\Lambda;\, \exists\, \gamma \in  \Gamma^e(\sigma),\,0\in V(\gamma)\right\}\right)\leq \sum_{\gamma;\, 0\in V(\gamma)}\mu_{\Lambda,\beta}^r\left(\left\{\sigma \in \Omega^r_\Lambda; \gamma\in\Gamma^e(\sigma)\right\}\right).
$$
Let  $\Omega(\gamma)=\{\sigma \in \Omega^r_\Lambda;\, \gamma\in \Gamma^e(\sigma)\}$. Using Proposition \ref{BE}, we have $\mu^r_{\Lambda,\beta}(\Omega(\gamma)) \leq e^{-\beta c_2|\gamma|}$.
Then,
\begin{align*}
\mu_{\Lambda,\beta}^r(\sigma_0\neq r)&\leq \sum_{\gamma;\, 0\in V(\gamma)} {e}^{-\beta c_2|\gamma|} = \sum_{n\geq 1} {e}^{-\beta c_2n} |\{\gamma;\, 0\in V(\gamma),\, |\gamma|=n\}|.
\end{align*}
By Proposition \ref{Bound_on_C_0_n}, 
$$
\mu_{\Lambda,\beta}^r(\sigma_0\neq r)\leq \sum_{n\geq 1}{e}^{-\beta c_2n}.e^{\left(c_1+\log q\right)n}\\[0.2cm]
= \sum_{n\geq 1}{e}^{-\left(\beta c_2 -c_1-\log{q} \right)n}\\[0.2cm]
= \frac{{e}^{-\left(\beta c_2 -c_1-\log{q} \right)}}{1-{e}^{-\left(\beta c_2 -c_1-\log{q} \right)}}.
$$
Then,
$$
\mu_{\Lambda,\beta}^r(\sigma_0 =  r)\leq 1 - \frac{{e}^{-\left(\beta c_2 -c_1-\log{q} \right)}}{1-{e}^{-\left(\beta c_2 -c_1-\log{q} \right)}},
$$
which goes to $1$ when $\beta \to \infty$.
\end{proof}
\begin{proof}[Proof of Corollary \ref{corol_phase_transition}]
By the proposition above, there is $\beta_0$ such that
\begin{equation}\label{phase_transition}
\mu_{\Lambda,\beta}^r(\sigma_0\neq r)<\frac{1}{4}, \forall \beta > \beta_0, \forall r \in \Z_q.
\end{equation}

Since $\mu_{\Lambda,\beta}^r(\sigma_0=r)+ \mu_{\Lambda,\beta}^r(\sigma_0\neq r)=1$, we have that $\mu_{\Lambda,\beta}^r(\sigma_0= r)>\frac{3}{4}$.

By Equation \eqref{phase_transition} and taking the thermodynamic limit (which exists by Corollary \ref{thermo_limit}), $\mu^\ell_\beta(\sigma_0 = r) \leq \mu^\ell_\beta(\sigma_0 \neq \ell) \leq 1/4$, while $\mu^r_\beta(\sigma_0 = r) \geq 3/4$.
\end{proof}

\begin{remark}
    We can take  $\beta_0 = (c_1+\ln{5}+\log{q})/c_2$.
\end{remark}

\section{Applications: deterministic and random perturbations.}\label{application}

As an example of the robustness of our methods for proving phase transition, this section will present the occurrence of phase transition for the Potts model in the presence of a random or decaying field as an application. 

\subsection{Decaying field}

The Hamiltonian of the Potts model with a general external field can be written as follows.

\begin{equation}
    H^q_{\Lambda, h}(\sigma) = -\sum_{\{x, y\} \subset \Lambda} J_{xy}\mathbbm{1}_{\{\sigma_x = \sigma_y\}} - \sum_{\substack{x \in \Lambda \\ y \in \Lambda^c}} J_{xy}\mathbbm{1}_{\{\sigma_x = q\}} - \sum_{x \in \Lambda} h_{x, \sigma_x},
\end{equation}
where $h = (h_{x, n})_{\substack{x \in \Z^d\\ n \in \Z_q}}$ is a family of non-negative real numbers. 

\begin{proof}[Proof of Theorem \ref{thm:potts_decaying_long}]
    Let $\sigma \in \Omega^q_{\Lambda}$ be any configuration and $\gamma \in \Gamma^{e}(\sigma)$. Define as before $\tau : = \tau_{\gamma}(\sigma)$. Using Proposition \ref{BE}, we have
    \begin{align*}
        H^q_{\Lambda, h}(\sigma) - H^q_{\Lambda, h}(\tau) &= H^q_{\Lambda}(\sigma) - H^q_{\Lambda}(\tau)  - \left(\sum_{x \in \Sp(\gamma) \cup \I'(\gamma)} h_{x, \sigma_x} - h_{x, \tau_x} \right) \\
        &\geq c_2\left(|\gamma| +  F_{\Sp(\gamma)} + \sum_{n = 1}^{q-1}F_{\I_n(\gamma)}\right) - \sum_{x \in \Sp(\gamma) \cup \I'(\gamma)} h_{x, \sigma_x} \\
        &=  \left(c_2 |\gamma| - \sum_{x \in \Sp(\gamma)} h_{x, \sigma_x} \right) + \sum_{n = 1}^{q-1}\left( c_2 F_{\I_n(\gamma)} - \sum_{x \in \I_n(\gamma)} h_{x, \sigma_x} \right).
    \end{align*}

Proceeding similarly to \cite{Affonso.2021}, we refer to the Theorem 7.33 of \cite{georgii.gibbs.measures}, which allows us to replace the original field by a truncated one given by
\begin{equation*}
    \widehat{h}_{x, n} = 
    \begin{cases}
        h_{x, n} &\text{ if } |x| \geq R,\\[0.2cm]
        0 &\text{ if } |x| < R,
    \end{cases}
\end{equation*}
where $R$ will be chosen later, without compromising the existence (or not) of the phase transition. Notice that, by Equation \eqref{field_hypothesis},

\begin{equation}\label{bound_noh}
    \sum_{x \in \Lambda} \widehat{h}_{x, n} \leq \frac{h^{\ast}|\Lambda|}{R^\delta},
\end{equation}
so that $R^\delta > 2h^{\ast}/c_2$ gives us 

\begin{equation*}
    c_2 |\gamma| - \sum_{x \in \Sp(\gamma)} h_{x, \sigma_x} \geq \frac{c_2}{2}|\gamma|.
\end{equation*}

Now, using again Equation \eqref{field_hypothesis}, the only remaining thing to be shown is that, for any finite subset $\Lambda \Subset \Z^d$, 

\begin{equation}\label{inequality_decaying}
    c_2 F_{\Lambda} - \sum_{x \in \Lambda} \widehat{h}_{x, n} \geq 0.
\end{equation}

This analysis was already performed in Proposition 4.7 from \cite{Affonso.2021}, and is guaranteed for $\delta > (\alpha - d) \wedge 1$ or $\delta = (\alpha - d) \wedge 1$ if $h^\ast$ is large enough.

\end{proof}

Although in this paper we have been mainly concerned with the long-range case, the methods developed here are also useful for the short-range one. Notice that the nearest-neighbor Potts model consists of the interactions given by \eqref{Long-Range Interaction} when $\alpha \to +\infty$, so it is natural to expect that the Theorem above also holds in this case for $\delta > 1$. The proof is very similar to the long-range case and the sketch of the proof is presented below.

\begin{proof}[Proof of Corollary \ref{thm:Potts_decaying_short}]
    The proof starts by following the same lines as the proof of Theorem \ref{thm:potts_decaying_long}. The difference is that, in the short-range case, we have $F_{\Lambda} = J|\partial\Lambda|$. A quick computation can shows us that we still have
    
    \begin{equation*}
         H^q_{\Lambda}(\sigma) - H^q_{\Lambda}(\tau) \geq c_2' \left(|\gamma| + \sum_{n = 1}^{q - 1}|\partial \I_n(\gamma)| \right),
    \end{equation*}
    for some constant $c_2'$. The unique inequality left to be proven is, thus,

\begin{equation*}
   c_2' |\partial \I_n(\gamma)| - \sum_{x \in \I_n(\gamma)} \widehat{h}_x \geq 0.
\end{equation*}

Now, notice that in the case $\alpha > d+1$, the analysis done for \eqref{inequality_decaying} in \cite{Affonso.2021} uses that $F_{\Lambda} \geq K|\partial \Lambda|$ for some constant $K > 0$, so the computation performed is exactly the same.

\end{proof}
\subsection{Random field}

The strategy will be to follow the idea presented in \cite{Cassandro.Picco.09, Ding2023}, where both the spins and the field are flipped in the Peierls argument. For such, we will introduce the joint distribution\footnote{Notice that the superscript $q$ in $\mathbb{R}^q$ refers to the Cartesian product, while the superscript $q$ in $\mathbb{Q}$ and $g$ reminds of the number of states and boundary condition.} $\mathbb{Q}^q$ on $\Omega \times (\mathbb{R}^q)^\Lambda$ defined, for any $\mathcal{A}\subset\Omega$ measurable and any Borel set $\mathcal{B}\subset (\mathbb{R}^q)^{\Lambda}$, as 

\begin{equation*}
    \mathbb{Q}_{\Lambda; \beta, \varepsilon}^q(\sigma \in \mathcal{A}, h_\Lambda\in \mathcal{B}) \coloneqq \sum_{\sigma \in \mathcal{A}}\int_{\mathcal{B}}  g_{\Lambda; \beta, \varepsilon}^q(\sigma, h_\Lambda) d h_\Lambda,
\end{equation*}
 where $dh_\Lambda$ is the product Lebesgue measure and, as before, $h_\Lambda=\{h_{x,n}\}_{x\in \Lambda, n\in \Z_q}$. The density being integrated is
\begin{equation*}
    g_{\Lambda; \beta, \varepsilon}^q(\sigma,h) \coloneqq \left[\prod_{x\in\Lambda}\frac{1}{2\pi^{\frac{q}{2}}}e^{-\frac{1}{2}\langle h_x, h_x \rangle}\right] \times \mu_{\Lambda;\beta, \varepsilon h}^q(\sigma),
\end{equation*}
where $\langle h_x, h_x \rangle = h_{x,1}^2+\dots+h_{x,q}^2$. The ideas of \cite{Ding2023} were successfully adapted to the long-range Ising model in \cite{Johanes}. In the case of the Potts model, however, we cannot proceed exactly as \cite{Johanes} since flipping the sign of the field does not erase it from the energy estimate. Instead, we will need to permute the field colors inside the interiors. In order to present the strategy in a nice way, we will need to introduce some concepts.

Firstly, notice that there is a bijection between our configuration space $\Omega$ and the set $G$ of all ordered partitions of $\mathbb{Z}^d$ containing $q$ elements given by $\sigma \mapsto A = (\sigma^{-1}(\{1\}), ..., \sigma^{-1}(\{q\}))$. Also, if we introduce in $\Omega$ an operation given by the sum in each coordinate, $(\sigma \cdot \omega)_x = \sigma_x + \omega_x$, it is not difficult to see that, in $G$, this operation must be defined by a kind of convolution so that, for any pair $A, B \in G$, the $n$-th element of the ordered partition $(A \ast B)$ is given by

\begin{equation*}
    (A \ast B)_n = \bigcup_{t \in \mathbb{Z}_q} A_t \cap B_{n-t},
\end{equation*}
in order for $(\Omega, \cdot)$ to be isomorphic to $(G, \ast)$. Now, given any set $M$, we can define a function $\theta: \Omega \times M^{\mathbb{Z}_q \times \mathbb{Z}^d} \to M^{\mathbb{Z}_q \times \mathbb{Z}^d}$ by $(\theta(\sigma, h))_{x, r} = h_{x, r + \sigma(x)}$. It is not difficult to see that $\theta$ is a group action. In what follows, we will take $M = \mathbb{R}$ and consider the induced action of $G$ instead. Given $A \in G$ and $h \in \mathbb{R}^{\mathbb{Z}_q \times \mathbb{Z}^d}$, the image of the action will be denoted by $\theta_A(h)$. When $A = (\I_1(\gamma), \dots, \I_{q-1}(\gamma), (\I'(\gamma))^c)$, we write $\theta_\gamma \coloneqq \theta_A$. With these definitions and
Proposition \ref{BE}, we have the following.
    \begin{align*}
        H_{\Lambda, \varepsilon h}^q(\sigma) - H_{\Lambda, \varepsilon \theta_{\gamma}(h)}^q(\tau_{\gamma}(\sigma))  &= H^q_{\Lambda}(\sigma) - H^q_{\Lambda}(\tau)  - \varepsilon\hspace{-0.6cm}\sum_{x \in \Sp(\gamma) \cup \I'(\gamma)} \left( h_{x, \sigma_x} - \theta_\gamma(h)_{x, \tau_x} \right) \\
        &\geq c_2\left(|\gamma| +  F_{\Sp(\gamma)} + \sum_{n = 1}^{q-1}F_{\I_n(\gamma)}\right) - \varepsilon\hspace{-0.2cm}\sum_{x \in \Sp(\gamma)} (h_{x, \sigma_x}-h_{x,q}).
    \end{align*}
    Thus, 
\begin{align}\label{Eq: quotient.of.gs}
    \frac{g_{\Lambda; \beta, \varepsilon}^q(\sigma, h)}{g_{\Lambda; \beta, \varepsilon}^q(\tau_{\gamma}(\sigma),\theta_{\gamma}(h))} 
    &\leq \exp{- \beta c_2 |\gamma| +\beta \varepsilon\sum_{x\in \Sp(\gamma)} (h_{x,\sigma_x}-h_{x,q})}\frac{Z_{\Lambda; \beta, \varepsilon}^q(\theta_{\gamma}(h))}{Z_{\Lambda; \beta, \varepsilon}^q(h)}.
\end{align}
 Define
\begin{equation}
\Delta_A(h) \coloneqq -\frac{1}{\beta}\log{\frac{Z_{\Lambda; \beta, \varepsilon}^q(h)}{Z_{\Lambda; \beta, \varepsilon}^q(\theta_A(h))}}.
\end{equation}

Similarly as before, when $A = (\I_1(\gamma), \dots, \I_{q-1}(\gamma), (\I'(\gamma))^c)$, we put $\Delta_\gamma(h):= \Delta_A(h)$, so we can write
\begin{align}\label{Eq: quotient.of.gs2}
    \frac{g_{\Lambda; \beta, \varepsilon}^q(\sigma, h)}{g_{\Lambda; \beta, \varepsilon}^q(\tau_{\gamma}(\sigma),\theta_{\gamma}(h))} 
    &\leq \exp{- \beta c_2 |\gamma| + \beta\Delta_\gamma(h) +\beta \varepsilon\sum_{x\in \Sp(\gamma)} (h_{x,\sigma_x}-h_{x,q})}.
\end{align}

The problem now is to estimate the probability of the two terms in the argument, which compete with $|\gamma|$, being too large. More precisely, we define two bad events, as

\begin{equation*}
    \mathcal{E}_0^c \coloneqq \left\{ \sup_{ \gamma\in\mathcal{C}_0} \frac{\left|\varepsilon\sum_{x\in \Sp(\gamma)}(h_{x, \sigma_x} - h_{x, q})\right|}{c_2|\gamma|} > \frac{1}{4}\right\}
\end{equation*}
and
$$\mathcal{E}_1^c\coloneqq \left\{\sup_{\substack{\gamma\in\mathcal{C}_0}} \frac{\Delta_{\gamma}(h)}{c_2|\gamma|} > \frac{1}{4}\right\}.$$

To do this, we first need an analogous to Lemma 4.1 of \cite{Ding2023}. We will denote by $\mathbb P$ the probability measure with respect to $\{h_{x,r}\}$ and by $\mathbb E$ the respective expectation.

\begin{lemma}\label{Lemma: Concentration.for.Delta.General}
    For any $A, A'\in G$ and $\lambda>0$, we have 
\begin{equation}\label{Eq: Tail.of.Delta_A}
    \mathbb{P}\left(|\Delta_A(h)| \geq \lambda \Big\vert h_{A_q}\right) \leq 2e^{\frac{-\lambda^2}{2q\varepsilon^2|A_q^c|}},
\end{equation}
\begin{equation}\label{Eq: Tail.of.the.diff.of.Deltas}
     \mathbb{P}\left(|\Delta_A(h) - \Delta_{A'}(h)|>\lambda \Big\vert h_{A_q \cap A'_q}\right) \leq  2e^{-\frac{{\lambda^2}}{{2q.d(A,A')^2}}},
\end{equation}
where $d(A,A') = \varepsilon\sum_{n=1}^{q-1}|A_n \Delta A'_n|^{1/2}$. Also, for any $\Lambda \Subset \mathbb{Z}^d$,
\begin{equation*}
    \mathbb{P}\left(\left|\varepsilon\sum_{x\in \Lambda}(h_{x, \sigma_x} - h_{x, q})\right| \geq \lambda \Big\vert h_{\Lambda^c}\right) \leq 2e^{\frac{-\lambda^2}{2\varepsilon^2|\Lambda|}}
\end{equation*}

\end{lemma}
 \begin{proof}
Since the distribution of the variables $h_{x,n}$ are permutation invariant, we get that $\mathbb{E}(\Delta_A|h_{A_q}) = 0.$ Moreover, for any $h_{x,r}$, for $x\in A_q^c$ and $r\in \Z_q$, we get
\[
\left|\frac{\partial}{\partial h_{x,r}}\Delta_A(h)\right| = \varepsilon \left|\mu^q_{\beta,\Lambda,\varepsilon \theta_A(h)}\left(\sigma_x = r - \mathfrak{n}_A(x)\right)-\mu_{\beta,\Lambda,\varepsilon h}^q(\sigma_x = r)\right|\leq \varepsilon,
\]
where $\mathfrak{n}_A = \sum_{n \in \mathbb{Z}_q}n\mathbbm{1}_{\{x \in A_n\}}$. Then, $\|\nabla \Delta_A(h)\|^2_2 \leq q\varepsilon^2 |A_q^c|$. This bound together with the Gaussian concentration inequality due to Talagrand and Ledoux (See pages 10-12 of \cite{Ledoux.Talagrand.91}) implies \eqref{Eq: Tail.of.Delta_A}.

For the second estimate, since $\theta_A^q(h) = h$, we have that $(\theta_A(h),\theta_{A'}(h))$ is equal to $(\theta_{A}\circ\theta_{A'}^{q-1}(h), h)$ in distribution. By the fact that $\theta$ is an action, we have that $\theta_{A}\circ\theta_{A'}^{-1} = \theta_B$, where $B = A \ast(A')^{-1}$. This implies that $\Delta_A(h)-\Delta_{A'}(h)$ is equal in distribution to $\Delta_B(h)$, so arguments similar to those before yields
\[
\mathbb{P}\left(|\Delta_A(h) - \Delta_{A'}(h)|>\lambda\Big\vert h_{A_q \cap A'_q}\right) \leq 2e^{\frac{-\lambda^2}{2q\varepsilon^2 |B_q^c|}}.
\]
The proof of the second inequality is concluded by noticing that $B_q^c \subset \bigcup_{1\leq n\leq q-1} A_n \Delta A'_n$, which can be found using the explicit expression

\begin{equation*}
B_n = \bigcup_{t \in \mathbb{Z}_q} A_t \cap A'_{t- n},    
\end{equation*}

hence $\displaystyle |B_q^c|\leq \sum_{n=1}^{q-1}|A_n\Delta A'_n|\leq \left(\sum_{n=1}^{q-1}|A_n\Delta A'_n|^{1/2}\right)^2$.

The third inequality follows directly by the famous tail estimate
\begin{equation*}
\mathbb{P}(|X| > \lambda) \leq 2 e^{-\lambda^2/2\sigma^2},
\end{equation*}
for a random variable $X \sim \mathcal{N}(0, \sigma^2)$, and noticing that $\varepsilon\sum_{x\in \Lambda}(h_{x, \sigma_x} - h_{x, q}) \sim \mathcal{N}(0, 2\varepsilon^2|\Lambda|)$.


\end{proof}

The following proposition deals with the first bad event.

\begin{proposition}\label{bound_E0}
    For $\epsilon > 0$ small enough, there exists $C_0 = C_0(\alpha, d)$ such that $\mathbb{P}(\mathcal{E}_0^c)\leq e^{-\frac{C_0}{\varepsilon^2}}$.
\end{proposition}
\begin{proof}
By Lemma \ref{Lemma: Concentration.for.Delta.General} and Proposition \ref{Bound_on_C_0_n}, 
\begin{equation}\label{Eq: Bound.bad.event.0.eq.1}
    \begin{split}
    \mathbb{P}(\mathcal{E}_0^c) &\leq \sum_{n\geq 1}\mathbb{P}\left(\sup_{\gamma\in\mathcal{C}_0(n)} \varepsilon\left|\sum_{x\in \Sp(\gamma)}h_{x, \sigma_x} - h_{x, q}\right| > \frac{c_2}{4}n\right)  \\
    &\leq \sum_{n\geq 1}\sum_{\gamma\in \mathcal{C}_0(n)} \mathbb{P}\left(\varepsilon\left|\sum_{x\in \Sp(\gamma)}h_{x, \sigma_x} - h_{x, q}\right| > \frac{c_2}{4}n\right) \\
    &\leq 2\sum_{n\geq 1}\sum_{\gamma\in \mathcal{C}_0(n)} \exp\left(-\frac{c^2_2 n}{64\varepsilon^2}\right) \\
    &\leq 2\sum_{n\geq 1} \exp\left(-\left[\frac{c^2_2}{64\varepsilon^2} - c_1 - \log q\right]n\right) \\
\end{split}
\end{equation}
where in the third inequality we used Lemma \ref{Lemma: Concentration.for.Delta.General}, and in the last inequality we used Proposition \ref{Bound_on_C_0_n}. Taking $\varepsilon < c_2(128(c_1 + \log q))^{-1/2}$,

\begin{align*}
     \mathbb{P}(\mathcal{E}_0^c) &\leq 2\sum_{n\geq 1} e^{ -\frac{c_2^2}{128\varepsilon^2}n} \leq e^{ -\frac{c_2^2}{256\varepsilon^2}} ,
\end{align*}
where the last inequality follows taking $\varepsilon \leq c_2(128\log3)^{-1/2}$. We conclude our proof by choosing $C_0 = c^2_2/256$. We needed to take 

\begin{equation*}
    \varepsilon \leq \frac{c_2}{\sqrt{128\max\{c_1 + \log q, 2\log 3\}}}
\end{equation*}
\end{proof}



For the second bad event, the proof of $\mathbb{P}(\mathcal{E}^c_1)\leq e^{-\frac{C_1}{\varepsilon^2}}$ closely follows the arguments in \cite[Section 3]{Johanes}. Here we are going to outline the major steps and the required adjustments. The two main ingredients of the proof are $(a)$ a general result on Gaussian processes connecting the supremum of the process with a geometric quantity of the space where the process is defined (Theorem \ref{Theo: Theo_2.2.27_Talagrand}) and $(b)$ that this geometric quantity is linear with respect to the size of the contours (Proposition \ref{Prop: Bound.gamma_2}). In the first place, let us introduce the geometric quantity just mentioned.
\renewcommand{\d}{\mathrm{d}}
\begin{definition}
    Given a set $T$, a sequence $(\mathcal{A}_n)_{n\geq 0}$ of partitions of $T$ is \textit{admissible} when $|\mathcal{A}_n|\leq 2^{2^n}$ and $\mathcal{A}_{n+1}\preceq \mathcal{A}_n$ for all $n\geq 0$.
\end{definition}

Given $t\in T$ and an admissible sequence $(\mathcal{A}_n)_{n\geq 0}$, $A_n(t)$ denotes the element of $\mathcal{A}_n$ that contains $t$. 

\begin{definition}
    Given $\theta > 0$ and a metric space $(T,\d)$, we define
    \begin{equation*}
        \gamma_\theta(T,\d) \coloneqq \inf_{(\mathcal{A}_n)_{n\geq 0}}\sup_{t\in T}\sum_{n\geq 0}2^{\frac{n}{\theta}}\diam(A_n(t)),
    \end{equation*}
where the infimum is taken over all admissible sequences of partitions. 
\end{definition}

Now we are ready to state the first ingredient.

\begin{theorem}\label{Theo: Theo_2.2.27_Talagrand} Given a metric space $(T,\d)$ and a family $(X_t)_{t\in T}$ of centered random variables satisfying \begin{equation}\label{Eq: Sub_gaussian_def}
    \mathbb{P}\left( |X_t - X_s| \geq \lambda \right) \leq 2\exp{\frac{-\lambda^2}{2\d(s,t)^2}},
\end{equation}
there is a universal constant $L>0$ such that, for any $u>0$,
\begin{equation*}
\mathbb{P}\left( \sup_{t\in T}X_t > L(\gamma_2(T,\d) + u\diam(T)) \right)\leq e^{-{u^2}},
\end{equation*}
where the $\diam(T)$ is the diameter taken with respect to the distance $\d$
\end{theorem} 
A proof can be found in \cite[Theorem 2.2.27]{Talagrand_14}.

In order to apply this general result to our case, we need to define a suitable metric space. Taking as inspiration that choice made in \cite{Johanes}, we will take $T_n := \{(\I_1(\gamma), ..., \I_{q-1}(\gamma), (\I'(\gamma))^c); \gamma \in \mathcal{C}_0(n)\}$. In order to apply Lemma \ref{Lemma: Concentration.for.Delta.General} and Equation \eqref{Eq: Sub_gaussian_def} be satisfied, the metric must be as defined before, $\mathrm{d}(A, A') = \varepsilon\sum_{n=1}^{q-1}|A_n \Delta A'_n|^{1/2}$. 

The second ingredient is the following proposition.

\begin{proposition}\label{Prop: Bound.gamma_2}
    Given $n\geq 0$, $d\geq 3$ and $\alpha > d$, there is a constant $L_1 \coloneqq L_1(d,\alpha)>0$  such that $$\gamma_2(T_n,\d_2) \leq \varepsilon L_1 n.$$
\end{proposition}

This proposition will be proved later. For now, our task will be to show that this setup works properly to prove the desired bound:

\begin{proposition}\label{Prop: Bound.bad.event.1}     
    There exists $C_1\coloneqq C_1(\alpha, d)$ such that $\mathbb{P}(\mathcal{E}_1^c)\leq e^{-\frac{C_1}{\varepsilon^2}}$ for any $\varepsilon^2<C_1$. 
\end{proposition}

\begin{proof}   
 By the union bound,
\begin{align}\label{Eq: Union_bound_bad_event}
    \mathbb{P}\left({\sup_{\substack{\gamma\in\mathcal{C}_0}} \frac{\Delta_{\gamma}(h)}{c_2|\gamma|} > \frac{1}{4}}\right) \leq \sum_{n=2}^\infty \mathbb{P}\left({\sup_{\substack{\gamma\in\mathcal{C}_0(n)}} \Delta_{\gamma}(h) > \frac{c_2}{4}}|\gamma|\right). 
\end{align}
Let $\gamma,\gamma^\prime\in \mathcal{C}_0(n)$ be two contours satisfying $\diam(T_n) = \d_2[(\I_1(\gamma), ...,\I'(\gamma)^c ),(\I_1(\gamma'), ...,\I'(\gamma')^c)]$. By the isoperimetric inequality, $|\I_m(\gamma)| \leq n^{\frac{d}{d-1}}$ for any $m \in \mathbb{Z}_q$, so we have
\begin{equation*}
    \diam(T_n)= \varepsilon\sum_{m = 1}^{q-1}{|\I_m(\gamma)\Delta \I_m(\gamma^\prime)|}^{\frac{1}{2}} \leq \sqrt{2}(q-1)\varepsilon n^{(\frac{d}{d-1})\frac{1}{2}} = \sqrt{2}(q-1)\varepsilon n^{(\frac{1}{2} + \frac{1}{2(d-1)})}.
\end{equation*}
Together with Proposition \ref{Prop: Bound.gamma_2}, this yields
    \begin{align*}
  \frac{c_2}{4}|\gamma| &= L\left[\varepsilon L_1 n + \varepsilon L_1 \left(\frac{c_2}{4\varepsilon L_1 L} - 1\right)n\right]\\
    &\geq   L\left[\gamma_2(T_n,\d) +  \frac{C_1^\prime}{\varepsilon}n^{\frac{1}{2} - \frac{1}{2(d-1)}}\diam(T_n))\right],
\end{align*}

with $C_1^\prime = \frac{c_2}{8\sqrt{2}(q-1)L}$ and $\varepsilon<\frac{c_2}{8L_1L}$. Applying Theorem \ref{Theo: Theo_2.2.27_Talagrand} with $u = \frac{C_1^\prime}{\varepsilon}n^{\frac{1}{2} - \frac{1}{2(d-1)}}$, we have
\begin{align*}
    \mathbb{P}\left({\sup_{\substack{\gamma\in\mathcal{C}_0(n)}} \Delta_{\gamma}(h) > \frac{c_2}{4}}|\gamma|\right) &=  \mathbb{P}\left(\sup_{\gamma\in\mathcal{C}_0(n)} \Delta_{\gamma}(h) > \frac{c_2}{4}n\right) \\
    &\leq \mathbb{P}\left({\sup_{\substack{\gamma\in\mathcal{C}_0(n)}} \Delta_{\gamma}(h) >    L\left[\gamma_2(T_n,\d) +  \frac{C_1^\prime}{\varepsilon}n^{\frac{1}{2} - \frac{1}{2(d-1)}}\diam(T_n)\right]}\right) \\ 
    &\leq \exp\left\{ - \frac{C_1^{\prime2}n^{1 - \frac{1}{(d-1)}}}{\varepsilon^2}\right\}
    \end{align*}
Using this back in equation \eqref{Eq: Union_bound_bad_event}, we conclude that 
\begin{equation*}
       \mathbb{P}\left({\sup_{\substack{\gamma\in\mathcal{C}_0}} \frac{\Delta_{\I_-(\gamma)}(h)}{c_2|\gamma|} > \frac{1}{4}}\right) \leq \sum_{n=2}^\infty \exp\left\{ - \frac{C_1^{\prime2}n^{1 - \frac{1}{(d-1)}}}{\varepsilon^2}\right\} \leq e^{-\frac{C_1}{\varepsilon^2}},
\end{equation*}
for a suitable constant $C_1\coloneqq C_1(\alpha, d)$ smaller than $\frac{{C_1^\prime}^2}{2}$ and $\varepsilon< C_1$. The dependency on $\alpha$ is due to the dependency on $c_2(\alpha, d)$.

\end{proof}

Now, let us return our attention to the proof of Proposition \ref{Prop: Bound.gamma_2}. For such, we are going to need adaptations of Proposition 3.17, Corollary 3.19 and Proposition 3.30 from \cite{Johanes}. Both Proposition 3.17 and 3.30 from \cite{Johanes} are purely geometric and, although it is stated for contours in the Ising model, the proofs rely only on the fact that these contours have irreducible $(M,a)$-partitions as support, so they hold in our case \emph{without any modification}. The linkage between this geometric aspect of the contours and the metric space is provided by Corollary 3.19. Since our metric space is different, some minor modifications are needed and we chose to present here the proof for completeness. 

In the first place, we will need to introduce the concept of a $r\ell$-cube and of admissible cubes. A $m$-cube is defined by

\begin{equation}
    C_{m}(x) \coloneqq \left(\prod_{i=1}^d{\left[2^{m}x_i , \ 2^{m}(x_i+1) \right)}\right)\cap \Z^d,
\end{equation}

where $x \in \mathbb{Z}^d$ (see \cite[Section 2.2]{Johanes}). We will often drop the origin point $x$ and write simply $C_m$. In general, we will take $m = r\ell$, where $r\coloneqq 4\lceil\log_2(a+1) \rceil + d +1$, $\lceil x \rceil$ being the smallest integer greater than or equal to $x$, and $\ell$ will be a natural number reflecting the scale on a multiscale analysis. 

Given some subset $A \subset \mathbb{Z}^d$, a $r\ell$-cube $C_{r\ell}$ is called \emph{admissible} if more than a half of its points are inside $A$. The set of admissible cubes for $A$ is

\begin{equation*}
    \mathfrak{C}_\ell(A) := \left\{C_{r\ell}; |C_{r\ell} \cap A| \geq \frac{1}{2}|C_{r\ell}|\right\}.
\end{equation*}

We abbreviate $\mathfrak{C}^m_\ell(\gamma) := \mathfrak{C}_\ell(\I_m(\gamma))$. Finally, we put $B^m_\ell(\gamma) := \bigcup_{C \in \mathfrak{C}^m_\ell(\gamma)} C$ to denote all the region encompassed by the cubes in $\mathfrak{C}^m_\ell(\gamma)$ and $B_\ell(\gamma) := (B^1_\ell(\gamma), ..., B^{q-1}_\ell(\gamma))$.

\begin{proposition}[Adaptation of Corollary 3.19 from \cite{Johanes}]\label{corollary_adapted}
      There exists a constant $b_3>0$ such that, for any $\ell>0$ and any two contours $\gamma_1,\gamma_2 \in \mathcal{C}_0(n)$ with $B_\ell(\gamma_1)=B_{\ell}(\gamma_2)$, 
    \begin{equation*}
        \d((\I_1(\gamma_1), ..., \I_{q-1}(\gamma_1)),(\I_1(\gamma_2), ..., \I_{q-1}(\gamma_2)))\leq 4 \varepsilon (q-1)b_3 2^{\frac{r\ell}{2}} n^{\frac{1}{2}}. 
    \end{equation*} 
\end{proposition}
\begin{proof}
    Notice that $B_0(\gamma) = (\I_1(\gamma), ..., \I_{q-1}(\gamma))$. By a simple application of the triangular inequality,
    \begin{equation*}
        \d(B_0(\gamma_1), B_0(\gamma_2)) \leq \d(B_0(\gamma_1),B_\ell(\gamma_1)) + \d(B_\ell(\gamma_2),B_0(\gamma_2)).
    \end{equation*}
    Using the triangular inequality repeatedly, we have
    \begin{align*}
    \d(B_0(\gamma_1),B_\ell(\gamma_1)) &\leq \sum_{i=1}^\ell \d_2(B_{i-1}(\gamma_1),B_{i}(\gamma_1)) = \sum_{i=1}^\ell \varepsilon\sum_{n=1}^{q-1}\sqrt{|B^n_i(\gamma_1)\Delta B^n_{i-1}(\gamma_1)|} \\
        & \leq \varepsilon(q-1)\sqrt{b_2}\sqrt{n} \sum_{i=1}^\ell  2^{\frac{ir}{2}}   \leq 2\varepsilon(q-1)\sqrt{b_2}2^{\frac{r\ell}{2}} \sqrt{n} 
    \end{align*}
    where the bound for $|B^n_i(\gamma_1)\Delta B^n_{i-1}(\gamma_1)|$ used \cite[Proposition 3.17]{Johanes}. As the same bound holds for $d_2(B_0(\gamma_2),B_\ell(\gamma_2))$, the corollary is proved by taking $b_3 = 2\sqrt{b_2}$.
\end{proof}

Finally,

\begin{proof}[Proof of Proposition \ref{Prop: Bound.gamma_2}]
Using the Majorizing Measure Theorem \cite{Talagrand_87} and the Dudley's Entropy Bound \cite{Dudley67}, we get that there is a constant $\overline{L}$ such that

\begin{equation*}
    \gamma_2(T_n, d) \leq \overline{L} \int_0^{\infty} \sqrt{\log N(T_n, d, \epsilon)} d\epsilon,
\end{equation*}

where $N(T_n, d, \epsilon)$ is defined as the minimal number of balls with radius $\epsilon > 0$ necessary to cover the metric space $T_n$ using the metric $\d$.

 As $N(T_n, \d, \epsilon)$ is decreasing in $\epsilon$, we can bound the integral by a suitable series, getting
\begin{align*}
      \gamma_2(T_n, d) &\leq 4\varepsilon (q-1) b_3 \overline{L} n^{\frac{1}{2}} \sqrt{\log{N(T_n, d, 0)}} \\ 
      & \hspace{3cm}+ 4\varepsilon (q-1)b_3 \overline{L}n^{\frac{1}{2}}\sum_{\ell=0}^\infty (2^{\frac{r(\ell+1)}{2}} - 2^{\frac{r\ell}{2}})\sqrt{\log N(T_n, \d,4\varepsilon (q-1)b_3 2^{\frac{r\ell}{2}}n^{\frac{1}{2}})}.
\end{align*}
We can bound the first term by noticing that $N(T_n, d, 0) = |T_n|\leq |\mathcal{C}_0(n)|$. By Proposition \ref{Bound_on_C_0_n},
\begin{equation*}
    4\varepsilon (q-1)b_3 \overline{L} n^{\frac{1}{2}}\sqrt{\log{N(T_n, d, 0)}} \leq 4\varepsilon (q-1) b_3 \overline{L} (c_1 + \log q)^{\frac{1}{2}} n.
\end{equation*}

Since $\diam(T_n)\leq \sqrt{2}(q-1)\varepsilon n^{\frac{1}{2} + \frac{1}{2(d-1)}}$ (see the proof of Proposition \ref{Prop: Bound.bad.event.1} ), when $4\varepsilon (q-1)b_3 \overline{L} 2^{\frac{r\ell}{2}}n^{\frac{1}{2}}\geq \sqrt{2}(q-1)\varepsilon n^{\frac{1}{2} + \frac{1}{2(d-1)}}$, only one ball covers all interiors, hence all the terms in the sum above with $\ell > k(n)\coloneqq \lfloor{\frac{\log_{2^r}(n)}{(d-1)}\rfloor}$ are zero. By Proposition \ref{corollary_adapted}, we have

\begin{equation*}
    N(T_n, \d,4\varepsilon (q-1) b_3 2^{\frac{r\ell}{2}}n^{\frac{1}{2}})\leq |B_{\ell}(\mathcal{C}_0(n))|,
\end{equation*}

where $B_{\ell}(\mathcal{C}_0(n))$ is the image of $\mathcal{C}_0(n)$ by $B_\ell$, that is, the collection of all $(q-1)$-pairs of regions composed by $r\ell$-cubes which approximates $(\I_1(\gamma), ..., \I_{q-1}(\gamma))$. Following the exact same steps of \cite[Proposition 3.30]{Johanes}, we have that there are constants $c_4 = c_4(\alpha, d)$ and $\kappa = \kappa(\alpha, d)$ such that, for any $m \in \mathbb{Z}_q$,

\begin{equation*}
    |B^m_{\ell}(\mathcal{C}_0(n))| \leq \exp\left\{c_4 \frac{\ell^{\kappa + 1} n}{2^{r\ell(d-1)}}   \right\}.
\end{equation*}

Since $B_{\ell}(\mathcal{C}_0(n)) \subset B^1_{\ell}(\mathcal{C}_0(n)) \times \ldots \times B^{q-1}_{\ell}(\mathcal{C}_0(n))$,

\begin{equation*}
     N(T_n, \d,4\varepsilon (q-1) b_3 2^{\frac{r\ell}{2}}n^{\frac{1}{2}}) \leq |B_{\ell}(\mathcal{C}_0(n))| \leq \left[ \exp\left\{c_4 \frac{\ell^{\kappa + 1} n}{2^{r\ell(d-1)}}   \right\} \right]^{q-1} = \exp\left\{c'_4\frac{\ell^{\kappa + 1} n}{2^{r\ell(d-1)}}   \right\},
\end{equation*}

where $c'_4 = c_4(q-1)$. Putting everything together, we are left with

\begin{align*}
        \gamma_2(T_n, d) &\leq 4\varepsilon (q-1) b_3 \overline{L} (c_1 + \log q)^{\frac{1}{2}}n + 4\varepsilon (q-1)b_3 \overline{L} 2^{\frac{r}{2}}\sqrt{c'_4} n^{\frac{1}{2}}\sum_{\ell=1}^{k(n)}2^{\frac{r\ell}{2}}\sqrt{\frac{\ell^{\kappa + 1} n }{2^{r\ell(d-1)}}} \nonumber\\
        &\leq  4\varepsilon (q-1)b_3 \overline{L} 2^{\frac{r}{2}} \sqrt{c'_4}\left[ (c_1 + \log q)^{\frac{1}{2}} + \sum_{\ell=1}^{\infty}\left(\frac{\ell^\frac{\kappa+1}{2}}{2^{\frac{r\ell(d-2)}{2}}} \right)\right]n.
\end{align*}

The series above converges for any $d\geq 3$, and we conclude that 
\begin{equation*}
       \gamma_2(T_n, d)\leq \varepsilon L_1^\prime n,
\end{equation*}
with $L_1^\prime\coloneqq   4(q-1) b_3 \overline{L} 2^{\frac{r}{2}}\sqrt{c_4}\left[ (c_1 + \log 2)^{\frac{1}{2}} + \sum_{\ell=1}^{\infty}\left(\frac{\ell^\frac{\kappa+1}{2}}{2^{\frac{r\ell(d-2)}{2}}} \right)\right]$. 


\end{proof}

Putting the bounds on the two bad events together, we have

\begin{theorem}\label{Theo: Transicao_de_fase}
For $d\geq 3$ and $\alpha>d$, there exists a constant $C\coloneqq C(d,\alpha)$ such that, for $\beta$ large enough and $\varepsilon$ small enough, the event 
    \begin{equation}\label{Eq: PTLR}
        \mu_{\Lambda; \beta, \varepsilon h}^q(\sigma_0 \neq q) \leq e^{-C\beta} + e^{-C/\varepsilon^2} 
    \end{equation}
    has $\mathbb{P}$-probability bigger than $1 - e^{-C\beta} - e^{-C/\varepsilon^2}$.
\end{theorem}

\begin{proof}[Proof of Theorem \ref{Theo: Transicao_de_fase}]
        The proof is an application of the Peierls' argument, but now on the joint measure $\mathbb{Q}$. Define $\mathcal{E} = \mathcal{E}_0 \cap \mathcal{E}_1$. Using Proposition \ref{bound_E0} and the bound for $\mathbb{P}(\mathcal{E}_1^c)$,
        \begin{align}\label{Eq: Upper.bound.on.Q.1}
            \mathbb{Q}_{\Lambda; \beta, \varepsilon}^+(\sigma_0 \neq q) &=  \mathbb{Q}_{\Lambda; \beta, \varepsilon}^+(\{\sigma_0 \neq q\} , \mathcal{E}_0) + \mathbb{Q}_{\Lambda; \beta, \varepsilon}^+(\{\sigma_0 \neq q\}, \mathcal{E}_0^c) \nonumber \\
            & \leq \mathbb{Q}_{\Lambda; \beta, \varepsilon}^+(\{\sigma_0 \neq q\} , \mathcal{E}_0) +  e^{-C_0/\varepsilon^2} \nonumber \\
            & \leq \mathbb{Q}_{\Lambda; \beta, \varepsilon}^+(\{\sigma_0 \neq q\} , \mathcal{E}) + \mathbb{Q}_{\Lambda; \beta, \varepsilon}^+(\{\sigma_0 \neq q\}, \mathcal{E}_0 \cap \mathcal{E}_{1}^c)  + e^{-C_0/\varepsilon^2} \nonumber \\
            & \leq \mathbb{Q}_{\Lambda; \beta, \varepsilon}^+(\{\sigma_0 \neq q\} , \mathcal{E}) + e^{-C_1/\varepsilon^2}  + e^{-C_0/\varepsilon^2},
        \end{align}
 When $\sigma_0 \neq q$, there must exist a contour $\gamma$ with $0\in V(\gamma)$, hence
\begin{align}\label{Eq: Upper.bound.on.Q.2}
    \mathbb{Q}_{\Lambda; \beta, \varepsilon}^q(\{\sigma_0 \neq q\} , \mathcal{E}) &= \int_{\mathcal{E}}\sum_{\sigma : \sigma_0 \neq q}g_{\Lambda; \beta, \varepsilon}^q(\sigma, h)dh \nonumber \\
    &\leq  \sum_{\gamma: 0 \in V(\gamma)} \int_{\mathcal{E}}\sum_{\sigma; \gamma \in \Gamma^e(\sigma)}g_{\Lambda; \beta, \varepsilon}^q(\sigma, h)dh \nonumber \\
    &\leq  \sum_{\gamma: 0 \in V(\gamma)} \int_{\mathcal{E}}\sup_{\sigma; \gamma \in \Gamma^e(\sigma)}\frac{g_{\Lambda; \beta, \varepsilon}^q(\sigma, h)}{g_{\Lambda; \beta, \varepsilon}^q(\tau_{\gamma}(\sigma), \theta_\gamma(h))} \prod_{x\in\Lambda}\frac{1}{2\pi^{\frac{q}{2}}}e^{-\frac{1}{2}\langle h_x, h_x \rangle}  dh \\
    &=  \sum_{\gamma: 0 \in V(\gamma)} \frac{1}{(2\pi)^{q|\Lambda|/2}} \int_{\mathcal{E}}\sup_{\sigma; \gamma \in \Gamma^e(\sigma)}\frac{g_{\Lambda; \beta, \varepsilon}^q(\sigma, h)}{g_{\Lambda; \beta, \varepsilon}^q(\tau_{\gamma}(\sigma), \theta_\gamma(h))}e^{-\frac{1}{2}\sum_{x \in \Lambda}\langle h_x, h_x \rangle}  dh 
\end{align}

In the third equation, we used that $\sum_{\sigma; \gamma \in \Gamma^e(\sigma)}g_{\Lambda; \beta, \varepsilon}^q(\tau_{\gamma}(\sigma), \theta_\gamma(h)) \leq \prod_{x\in\Lambda}\frac{1}{2\pi^{\frac{q}{2}}}e^{-\frac{1}{2}\langle h_x, h_x \rangle} $. Equation \eqref{Eq: quotient.of.gs} implies, 
\begin{align}\label{Eq: Upper.bound.on.Q.3}
    \sup_{\sigma; \gamma \in \Gamma^e(\sigma)}\frac{g_{\Lambda; \beta, \varepsilon}^q(\sigma, h)}{g_{\Lambda; \beta, \varepsilon}^q(\tau_{\gamma}(\sigma), \theta_\gamma(h))} &\leq    e^{- \beta c_2 |\gamma| + \beta \Delta_{\gamma}(h)} \sup_{\sigma; \gamma \in \Gamma^e(\sigma)}e^{\beta \varepsilon\sum_{x\in \Sp(\gamma)} (h_{x,\sigma_x}-h_{x,q})}\nonumber\\
    &= e^{- \beta c_2 |\gamma| + \beta \Delta_{\gamma}(h) + \beta \varepsilon\sum_{x\in \Sp(\gamma)} (h_{x,\sigma_x}-h_{x,q})}\nonumber\\
    &\leq  e^{- \beta \frac{c_2}{2} |\gamma|},
\end{align}
since the configuration is fixed inside the contour $\gamma$ and $\Delta_{\gamma}(h) + \varepsilon\sum_{x\in \Sp(\gamma)} (h_{x,\sigma_x}-h_{x,q})\leq \frac{c_2}{2}|\gamma|$, for all $h\in\mathcal{E}$, thus

\begin{equation*}
    \frac{1}{(2\pi)^{q|\Lambda|/2}} \int_{\mathcal{E}}\sup_{\sigma; \gamma \in \Gamma^e(\sigma)}\frac{g_{\Lambda; \beta, \varepsilon}^q(\sigma, h)}{g_{\Lambda; \beta, \varepsilon}^q(\tau_{\gamma}(\sigma), \theta_\gamma(h))}e^{-\frac{1}{2}\sum_{x \in \Lambda}\langle h_x, h_x \rangle}  dh \leq e^{- \beta \frac{c_2}{2} |\gamma|}.
\end{equation*}

The inequality above, together with Equation \eqref{Eq: Upper.bound.on.Q.1} and Proposition \ref{Bound_on_C_0_n}, yields
\begin{align*}
     \mathbb{Q}_{\Lambda; \beta, \varepsilon}^q(\sigma_0 \neq q) & \leq  \sum_{\gamma: 0 \in V(\gamma)} e^{- \beta \frac{c_2}{2} |\gamma|} + e^{-C_0/\varepsilon^2} + e^{-C_1/\varepsilon^2} \\
     &\leq \sum_{n\geq 1}e^{(-\beta \frac{c_2}{2}+c_1+\log(q))n} + e^{-C_0/\varepsilon^2} + e^{-C_1/\varepsilon^2}\\
    &\leq \frac{e^{-\beta \frac{c_2}{4}}}{1-e^{-\beta \frac{c_2}{4}}} + e^{-C_0/\varepsilon^2} + e^{-C_1/\varepsilon^2}\\[0.3cm]
    & \leq e^{-\beta\frac{c_2}{8}} + e^{-C'/\varepsilon^2}.
\end{align*}
The last inequality holds for $C' = \min\{C_0, C_1\}/2$, $\varepsilon^2 \leq C'/\log 2$ and $\beta > 4/c_2$. Putting also $2C = \min\{c_2/8, C'\}$,
\begin{equation*}
    \mathbb{Q}_{\Lambda; \beta, \varepsilon}^q(\sigma_0 \neq q) \leq e^{-\beta 2C} + e^{-2C / \varepsilon^2}.
\end{equation*}
The Markov Inequality finally yields
\begin{align*}
    \mathbb{P}\left( \mu_{\Lambda; \beta, \varepsilon h}^q(\sigma_0 \neq q) \geq e^{-C\beta} + e^{-C/\varepsilon^2}\right) &\leq \frac{\mathbb{Q}_{\Lambda; \beta, \varepsilon}^q(\sigma_0 \neq q)}{e^{-C\beta} + e^{-C/\varepsilon^2}}  \leq e^{-C\beta} + e^{-C/\varepsilon^2},
\end{align*}
which proves our claim.
\end{proof}

\section{Concluding Remarks}\label{concluding}

In the proof of phase transition, only part of Proposition 4.3 was used --- specifically, the requirement that the difference between the Hamiltonians is bounded below by a quantity proportional to $|\gamma|$. However, the full estimate is critical for further applications, such as establishing the convergence of the cluster expansion at low temperatures, this is certainly achievable for the models addressed in this paper through a straightforward adaptation of the methods in \cite{cluster}.

Another natural direction for future research is improving the results by Park \cite{Park.88.I, Park.88.II} on the Pirogov-Sinai theory for ferromagnetic long-range interactions to encompass all $\alpha > d$. Currently, his results apply only to $\alpha > 3d + 1$. We aim to address this limitation in subsequent papers, which are already in preparation.

The phase diagram of Ising and Potts models with decaying fields remains incomplete, even in the short-range case. A key open question is whether uniqueness holds for the critical exponent $\delta = 1$ when $h^*$ is sufficiently large (see \cite{Bis2}). Proving uniqueness when the field decays slowly is nonstandard, with the only known argument combining results from \cite{Bis2} and \cite{Cioletti2015} for the nearest-neighbor Ising case. Since the uniqueness of the Gibbs state in the short-range setting was obtained via contour arguments (see \cite{Bis2}), it is natural to explore whether similar arguments extend to disconnected contours in long-range models. Phase transitions for these models were established using the Peierls argument in the one-dimensional case \cite{Bissacot.Endo.18}, and for what appears to be a sharp region of exponents in the multidimensional case \cite{Affonso.2021} --- the same region covered in this paper.

Our proof of phase transition in the presence of a random field closely follows the methods in \cite{Johanes, Ding2023}, which demands $d \geq 3$. A thorough discussion of how the multiscaled contours present themselves in $d = 1$ for $\alpha > d$ (and not only when $\alpha = 2$ as in \cite{Frohlich.Spencer.82}) is addressed in \cite{Affonso_Bissacot_Corsini_Welsch_2024}, dealing also with the decaying field. For results on the random field long-range Ising model with $d = 1$ and $2$, see \cite{joao_chineses}.

In summary, many other results for short-range $q$-state models rely on the notion of contours. We expect that most of these results can be extended to the long-range setting using the multiscaled contours, tools, and ideas developed here.

\section*{Acknowledgements}

The authors are very grateful to Pierre Picco for useful comments about one-dimensional contours for long-range systems during the workshop \textit{Randomness 2024} at the Institute of Mathematics and Statistics (IME-USP) in February 2024. RB thanks Aernout van Enter for his generosity in sharing his insights and his knowledge with the community; all the discussions over the years were very important for the Brazilian group working in Mathematical Physics at the University of São Paulo (USP); the authors also thanks to him and Roberto Fernández for sharing many references from the area, in particular, the monograph by Gruber, Hintermann and Merlini \cite{Gruber}, which is one of the starting points of our work. The authors also thank João Maia for his suggestions concerning the random field case, for his careful reading in the preliminary version of this paper, and for his friendship during all these years. RB and GF are supported by the DINTER Project, a cooperation agreement between USP and IFMT, via the Graduate Program of Applied Mathematics at IME-USP.  RB was partially supported by USP-COFECUB Uc Ma 176/19, ``{\it Formalisme Thermodynamique des quasi-cristaux \`a temp\'erature z\'ero}\rq\rq. This study was supported by the São Paulo Research Foundation (FAPESP), Brasil, Processes Numbers 2016/25053-8 and 2023/00854-1. KW is supported by CAPES and CNPq grant 160295/2024-6. RB is supported by CNPq grants 311658/2025-3 and 407527/2025-7.

\bibliographystyle{habbrv}
\bibliography{refs}

\begin{thebibliography}{10}
\expandafter\ifx\csname url\endcsname\relax
  \def\url#1{\texttt{#1}}\fi
\expandafter\ifx\csname doi\endcsname\relax
  \def\doi#1{\burlalt{doi:#1}{http://dx.doi.org/#1}}\fi
\expandafter\ifx\csname urlprefix\endcsname\relax\def\urlprefix{URL }\fi
\expandafter\ifx\csname href\endcsname\relax
  \def\href#1#2{#2}\fi
\expandafter\ifx\csname burlalt\endcsname\relax
  \def\burlalt#1#2{\href{#2}{#1}}\fi

\bibitem{large_deviations}
L.~Affonso and R.~Bissacot.
\newblock Second order large deviation bounds for long-range Ising models.
\newblock \emph{In preparation}, 2024.

\bibitem{Affonso_Bissacot_Corsini_Welsch_2024}
L.~Affonso, R.~Bissacot, H.~Corsini, and K.~Welsch.
\newblock Phase Transitions on 1d Long-Range Ising Models with Decaying Fields: A Direct Proof via Contours, 2024, \burlalt{ArXiv:2412.07098}{http://arxiv.org/abs/2412.07098}.

\bibitem{Affonso.2021}
L.~Affonso, R.~Bissacot, E.~O. Endo, and S.~Handa.
\newblock Long-range Ising models: Contours, phase transitions and decaying fields.
\newblock {\em Journal of the European Mathematical Society}, 27(4):1679--1714, 2025.

\bibitem{Johanes}
L.~Affonso, R.~Bissacot, and J.~Maia.
\newblock Phase Transitions in Multidimensional Long-Range Random Field Ising Models, 2024, \burlalt{ArXiv:2307.14150}{http://arxiv.org/abs/2307.14150}.

\bibitem{cluster}
L.~Affonso, R.~Bissacot, J.~Maia, J.~Rodrigues, and K.~Welsch.
\newblock Cluster Expansion and Decay of Correlations for Multidimensional Long-Range Ising Models, 2025, \burlalt{ArXiv:2508.15666}{http://arxiv.org/abs/2508.15666}.

\bibitem{Aizenman1988}
M.~Aizenman, J.~T. Chayes, L.~Chayes, and C.~M. Newman.
\newblock Discontinuity of the magnetization in one-dimensional $1/|x-y|^2$ Ising and Potts models.
\newblock {\em Journal of Statistical Physics}, 50(1–2):1–40, 1988.

\bibitem{velenik_field}
Y.~Aoun, S.~Ott, and Y.~Velenik.
\newblock Fixed-magnetization Ising model with a slowly varying magnetic field.
\newblock {\em Journal of Statistical Physics}, 191(10):125, 2024.

\bibitem{menino_maluqinho}
A.~Bakchich, A.~Benyoussef, and L.~Laanait.
\newblock Phase diagram of the {Potts} model in an external magnetic field.
\newblock {\em Annales de l'I.H.P. Physique th\'eorique}, 50(1):17--35, 1989.

\bibitem{Beffara2011}
V.~Beffara and H.~Duminil-Copin.
\newblock The self-dual point of the two-dimensional random-cluster model is critical for $q \geq 1$.
\newblock {\em Probability Theory and Related Fields}, 153(3–4):511–542, 2011.

\bibitem{Biskup2000_PRL}
M.~Biskup, C.~Borgs, J.~T. Chayes, L.~J. Kleinwaks, and R.~Kotecký.
\newblock General Theory of Lee-Yang Zeros in Models with First-Order Phase Transitions.
\newblock {\em Physical Review Letters}, 84(21):4794–4797, 2000.

\bibitem{Biskup2000}
M.~Biskup, C.~Borgs, J.~T. Chayes, and R.~Kotecký.
\newblock Gibbs states of graphical representations of the Potts model with external fields.
\newblock {\em Journal of Mathematical Physics}, 41(3):1170–1210, 2000.

\bibitem{Biskup2003}
M.~Biskup and L.~Chayes.
\newblock Rigorous Analysis of Discontinuous Phase Transitions via Mean-Field Bounds.
\newblock {\em Communications in Mathematical Physics}, 238(1):53–93, 2003.

\bibitem{Biskup2006}
M.~Biskup, L.~Chayes, and N.~Crawford.
\newblock Mean-Field Driven First-Order Phase Transitions in Systems with Long-Range Interactions.
\newblock {\em Journal of Statistical Physics}, 122(6):1139–1193, 2006.

\bibitem{Biskup2006_rp}
M.~Biskup and R.~Kotecký.
\newblock Forbidden Gap Argument for Phase Transitions Proved by Means of Chessboard Estimates.
\newblock {\em Communications in Mathematical Physics}, 264(3):631–656, 2006.

\bibitem{Bis2}
R.~Bissacot, M.~Cassandro, L.~Cioletti, and E.~Presutti.
\newblock Phase Transitions in Ferromagnetic {I}sing Models with Spatially Dependent Magnetic Fields.
\newblock {\em Communications in Mathematical Physics}, 337:41--53, 2015.

\bibitem{Bis1}
R.~Bissacot and L.~Cioletti.
\newblock Phase Transition in Ferromagnetic {I}sing Models with Non-uniform External Magnetic Fields.
\newblock {\em Journal of Statistical Physics}, 139:769--778, 2010.

\bibitem{BEE}
R.~Bissacot, E.~O. Endo, and A.~C. van Enter.
\newblock Stability of the phase transition of critical-field {I}sing model on Cayley trees under inhomogeneous external fields.
\newblock {\em Stochastic Processes and their Applications}, 127(12):4126--4138, 2017.

\bibitem{Bissacot.Endo.18}
R.~Bissacot, E.~O. Endo, A.~C.~D. van Enter, B.~Kimura, and W.~M. Ruszel.
\newblock Contour methods for long-range Ising models: weakening nearest-neighbor interactions and adding decaying fields.
\newblock {\em Annales Henri Poincar{\'e}}, 19:2557--2574, 2018.

\bibitem{bjornberg}
J.~E. Björnberg.
\newblock Graphical representations of Ising and Potts models. Doctoral Thesis Stockholm, Sweden, 2009, \burlalt{ArXiv:1011.2683}{http://arxiv.org/abs/1011.2683}.

\bibitem{Bogachev2019}
L.~V. Bogachev and U.~A. Rozikov.
\newblock On the uniqueness of Gibbs measure in the Potts model on a Cayley tree with external field.
\newblock {\em Journal of Statistical Mechanics: Theory and Experiment}, 2019(7):073205, 2019.

\bibitem{Bricmont.Kupiainen.88}
J.~Bricmont and A.~Kupiainen.
\newblock Phase transition in the 3d random field Ising model.
\newblock {\em Communications in Mathematical Physics}, 116:539--572, 1988.

\bibitem{Bricmont1985}
J.~Bricmont, K.~Kuroda, and J.~L. Lebowitz.
\newblock First order phase transitions in lattice and continuous systems: Extension of Pirogov-Sinai theory.
\newblock {\em Communications in Mathematical Physics}, 101(4):501–538, 1985.

\bibitem{Cassandro.Picco.09}
M.~Cassandro, E.~Orlandi, and P.~Picco.
\newblock Phase Transition in the 1d Random Field Ising Model with long range interaction.
\newblock {\em Communications in Mathematical Physics}, 288:731--744, 2009.

\bibitem{Cioletti2015}
L.~Cioletti and R.~Vila.
\newblock Graphical Representations for Ising and Potts Models in General External Fields.
\newblock {\em Journal of Statistical Physics}, 162(1):81–122, 2015.

\bibitem{Coquille2013}
L.~Coquille, H.~Duminil-Copin, D.~Ioffe, and Y.~Velenik.
\newblock On the Gibbs states of the noncritical Potts model on $\mathbb Z ^2$.
\newblock {\em Probability Theory and Related Fields}, 158(1–2):477–512, 2013.

\bibitem{joao_chineses}
J.~Ding, F.~Huang, and J.~Maia.
\newblock Phase transitions in low-dimensional long-range random field Ising models, 2024.
\newblock \doi{10.48550/ARXIV.2412.19281}.

\bibitem{Ding2023}
J.~Ding and Z.~Zhuang.
\newblock Long range order for random field Ising and Potts models.
\newblock {\em Communications on Pure and Applied Mathematics}, 77(1):37–51, 2023.

\bibitem{Dudley67}
R.~M. Dudley.
\newblock The sizes of compact subsets of Hilbert space and continuity of Gaussian processes.
\newblock {\em Journal of Functional Analysis}, 1:290--330, 1967.

\bibitem{DUMINILCOPIN2021}
H.~Duminil-Copin, M.~Gagnebin, M.~Harel, I.~Manolescu, and V.~Tassion.
\newblock Discontinuity of the phase transition for the planar random-cluster and Potts models with $q \geq 4$.
\newblock {\em Annales scientifiques de l’École Normale Supérieure}, 54(6):1363–1413, 2021.

\bibitem{DuminilCopin2016}
H.~Duminil-Copin, V.~Sidoravicius, and V.~Tassion.
\newblock Continuity of the Phase Transition for Planar Random-Cluster and Potts Models with ${1 \leq q \leq 4}$.
\newblock {\em Communications in Mathematical Physics}, 349(1):47–107, 2016.

\bibitem{folland}
G.~Folland.
\newblock {\em A Course in Abstract Harmonic Analysis}.
\newblock Studies in Advanced Mathematics. Taylor \& Francis, 1994.

\bibitem{Frohlich.Spencer.82}
J.~Fr{\"o}hlich and T.~Spencer.
\newblock The Phase Transition in the one-dimensional Ising model with $1/r^2$ interaction energy.
\newblock {\em Communications in Mathematical Physics}, 84:87--101, 1982.

\bibitem{georgii.gibbs.measures}
H.~Georgii.
\newblock {\em Gibbs Measures and Phase Transitions}.
\newblock De Gruyter studies in mathematics. De Gruyter, 2011.

\bibitem{Ginibre1970}
J.~Ginibre.
\newblock General formulation of Griffiths’ inequalities.
\newblock {\em Communications in Mathematical Physics}, 16(4):310–328, 1970.

\bibitem{Gobron2007}
T.~Gobron and I.~Merola.
\newblock First-Order Phase Transition in Potts Models with Finite-Range Interactions.
\newblock {\em Journal of Statistical Physics}, 126(3):507–583, 2007.

\bibitem{Goldschmidt1981}
Y.~Y. Goldschmidt.
\newblock Phase diagram of the Potts model in an applied field.
\newblock {\em Physical Review B}, 24(3):1374–1383, 1981.

\bibitem{Gruber}
C.~Gruber, A.~Hintermann, and D.~Merlini.
\newblock {\em Group analysis of classical lattice systems}.
\newblock Lecture Notes in Physics. Springer, Berlin, Germany, 1977 edition, 1977.

\bibitem{Imbrie.Newman.88}
J.~Z. Imbrie and C.~M. Newman.
\newblock An intermediate phase with slow decay of correlations in one dimensional {${1/|x- y|^2}$} percolation, Ising and Potts models.
\newblock {\em Communications in Mathematical Physics}, 118:303--336, 1988.

\bibitem{ising1925}
E.~Ising.
\newblock Beitrag zur theorie des ferromagnetismus.
\newblock {\em Zeitschrift f{\"u}r Physik}, 31(1):253--258, 1925.

\bibitem{Koteck1982}
R.~Kotecký and S.~B. Shlosman.
\newblock First-order phase transitions in large entropy lattice models.
\newblock {\em Communications in Mathematical Physics}, 83(4):493–515, 1982.

\bibitem{Laanait1991}
L.~Laanait, A.~Messager, S.~Miracle-Sole, J.~Ruiz, and S.~Shlosman.
\newblock Interfaces in the Potts model I: Pirogov-Sinai theory of the Fortuin-Kasteleyn representation.
\newblock {\em Communications in Mathematical Physics}, 140(1):81–91, 1991.

\bibitem{Ledoux.Talagrand.91}
M.~Ledoux and M.~Talagrand.
\newblock {\em Probability in Banach Spaces: Isoperimetry and Processes}.
\newblock A Series of Modern Surveys in Mathematics Series. Springer, 1991.

\bibitem{Lee-Yang.II.1952}
T.~D. Lee and C.~N. Yang.
\newblock Statistical Theory of Equations of State and Phase Transitions. II. Lattice Gas and Ising Model.
\newblock {\em Physical Review}, 87:410--419, 1952.

\bibitem{Martirosian1986}
D.~H. Martirosian.
\newblock Translation invariant Gibbs states in theq-state Potts model.
\newblock {\em Communications In Mathematical Physics}, 105(2):281–290, 1986.

\bibitem{Park.88.I}
Y.~M. Park.
\newblock Extension of Pirogov-Sinai theory of phase transitions to infinite range interactions I. Cluster expansion.
\newblock {\em Communications in Mathematical Physics}, 114:187--218, 1988.

\bibitem{Park.88.II}
Y.~M. Park.
\newblock Extension of Pirogov-Sinai theory of phase transitions to infinite range interactions. II. Phase diagram.
\newblock {\em Communications in Mathematical Physics}, 114:219--241, 1988.

\bibitem{Pearce1980}
P.~A. Pearce and R.~B. Griffiths.
\newblock Potts model in the many-component limit.
\newblock {\em Journal of Physics A: Mathematical and General}, 13(6):2143–2148, 1980.

\bibitem{Peierls.1936}
R.~Peierls.
\newblock On Ising's model of ferromagnetism.
\newblock In {\em Mathematical Proceedings of the Cambridge Philosophical Society}, volume~32, pages 477--481, 1936.

\bibitem{Pfister1991LargeDA}
C.~E. Pfister.
\newblock Large deviations and phase separation in the two-dimensional Ising model.
\newblock {\em Helvetica Physica Acta}, 64:953--1054, 1991.

\bibitem{Pirogov.Sinai.75}
S.~A. Pirogov and Y.~G. Sinai.
\newblock Phase diagrams of classical lattice systems.
\newblock {\em Teoreticheskaya i Matematicheskaya Fizika}, 25:358--369, 1975.

\bibitem{Pirogov1976}
S.~A. Pirogov and Y.~G. Sinai.
\newblock Phase diagrams of classical lattice systems, continuation.
\newblock {\em Teoreticheskaya i Matematicheskaya Fizika}, 26(1):39–49, 1976.

\bibitem{Potts1952}
R.~B. Potts.
\newblock Some generalized order-disorder transformations.
\newblock {\em Mathematical Proceedings of the Cambridge Philosophical Society}, 48(1):106–109, 1952.

\bibitem{Ray2020}
G.~Ray and Y.~Spinka.
\newblock A Short Proof of the Discontinuity of Phase Transition in the Planar Random-Cluster Model with $q \geq 4$.
\newblock {\em Communications in Mathematical Physics}, 378(3):1977–1988, 2020.

\bibitem{Rozikov2021}
U.~A. Rozikov.
\newblock {\em Gibbs Measures in Biology and Physics: The Potts Model}.
\newblock World Scientific, 2021.

\bibitem{Sinai2014-pe}
Y.~G. Sinai.
\newblock {\em Theory of Phase Transitions: Rigorous results}.
\newblock Elsevier, 2014.

\bibitem{Talagrand_87}
M.~Talagrand.
\newblock {Regularity of gaussian processes}.
\newblock {\em Acta Mathematica}, 159:99 -- 149, 1987.

\bibitem{Talagrand_14}
M.~Talagrand.
\newblock {\em Upper and lower bounds for stochastic processes}, volume~60.
\newblock Springer, 2014.

\bibitem{vanEnter1995}
A.~C.~D. van Enter, R.~Fernández, and R.~Kotecký.
\newblock Pathological behavior of renormalization-group maps at high fields and above the transition temperature.
\newblock {\em Journal of Statistical Physics}, 79(5–6):969–992, 1995.

\bibitem{Zahradnik.84}
M.~Zahradn{\'\i}k.
\newblock An alternate version of Pirogov-Sinai theory.
\newblock {\em Communications in Mathematical Physics}, 93:559--581, 1984.

\end{thebibliography}


\begin{thebibliography}{99}
	
	
	\bibitem{Bis1}R. Bissacot and L. Cioletti. \textit{Phase Transition in Ferromagnetic Ising Models with Non-uniform External Magnetic Fields}, J. Stat. Phys., \textbf{139}, 769-778 (2010).
	
	\bibitem{Bis2}R. Bissacot, M. Cassandro, L. Cioletti and E. Presutti. \textit{Phase Transitions in Ferromagnetic Ising Models with Spatially Dependent Magnetic Fields}, Commun. Math. Phys., \textbf{337}, 41-53, (2015).
	
	
	\bibitem{BEE} R. Bissacot, E. O. Endo, and A. C. D. van Enter. \textit{Stability of the phase transition of critical-
		field Ising model on Cayley trees under inhomogeneous external fields}, Stochastic Processes and their Applications \textbf{127}, 4126-4138, (2017).
	
	\bibitem{Eric2} R. Bissacot, E.O. Endo, A.C.D. van Enter, B. Kimura and W.M. Ruszel. \textit{Contour Methods for Long-Range Ising Models: Weakening Nearest-Neighbor Interactions and Adding Decaying Fields}, Ann. Henri Poincar\'{e}, \textbf{19}, 2557-2574 (2018).
	
	\bibitem{Bov} A. Bovier. \textit{Statistical Mechanics of Disordered Systems: A Mathematical Perspective}, Cambridge Series in Statistical and Probabilistic Mathematics, Cambridge University Press, (2006).
	
	\bibitem{BK} J. Bricmont and A. Kupiainen. \textit{Phase transition in the $3$d random field Ising model}. Commun. Math. Phys. \textbf{116}, 539-572, (1988).
	
	\bibitem{Cass} M. Cassandro, P.A. Ferrari, I. Merola and E. Presutti. \textit{Geometry of contours and Peierls estimates in $d = 1$ Ising models with long range interaction},  J. Math. Phys. \textbf{46}, 053305, (2005).
	
	\bibitem{Cass1} M. Cassandro, E. Orlandi and P. Picco. \textit{Phase Transition in the 1d Random Field Ising Model with Long Range Interaction},  Commun. Math. Phys., \textbf{288},731-744, (2009).
	
	\bibitem{Cass2} M. Cassandro, I. Merola, P. Picco and U. Rozikov. \textit{One-Dimensional Ising Models with Long Range Interactions: Cluster Expansion, Phase-Separating Point},  Commun. Math. Phys.  \textbf{327}, 951-991, (2014).
	
	\bibitem{Cass3} M. Cassandro, I. Merola and P. Picco. \textit{Phase Separation for the Long Range One-dimensional Ising Model},  J. Stat. Phys.  \textbf{167}, 351-382 (2017).
		
	\bibitem{CV} L. Cioletti and R. Vila.  \textit{Graphical Representations for Ising and Potts Models in General External Fields}. J. Stat. Phys., \textbf{162}, 81-122, (2016).

	
	\bibitem{Vel} S. Friedli and Y. Velenik. \textit{Statistical Mechanics of Lattice Systems: a Concrete Mathematical Introduction}, Cambridge University Press, (2017).
	
	\bibitem{Fro3} J. Fr\"{o}hlich and T. Spencer. \textit{The Kosterlitz-Thouless Transition in Two-Dimensional Abelian Spin Systems and the Coulomb Gas
	}, Commun. Math. Phys., \textbf{81}, 527-602, (1981).
	
	\bibitem{Fro1} J. Fr\"{o}hlich and T. Spencer. \textit{The Phase Transition in the one-dimensional Ising model with $1/r^2$ interaction energy}, Commun. Math. Phys.,\textbf{84}, 87-101, (1982).
	
	\bibitem{Fro2} J. Fr\"{o}hlich and T. Spencer. \textit{Absence of Diffusion in the Anderson Tight Binding model for Large Disorder or Low Energy}, Commun. Math. Phys. \textbf{88}, 151-184, (1983).
	
	
	\bibitem{Geo} H.O. Georgii. \textit{Gibbs measures and Phase Transitions}, De Gruyter Studies in Mathematics 9, Second Edition, (2009).
	
	\bibitem{GGR} J. Ginibre, A. Grossmann and D. Ruelle. \textit{Condensation of Lattice Gases}, Commun. Math. Phys. \textbf{3} 187-193, (1966).
	
	
	\bibitem{Imbr1} J.Z. Imbrie. \textit{Decay of Correlations in the One-Dimensional Ising Model with $J_{ij}=|i-j|^{-2}$}, Commun. Math. Phys. \textbf{85}, 491-515, (1982).
	
	\bibitem{Imbr2} J.Z. Imbrie and C.M. Newman. \textit{An intermediate Phase with Slow Decay of Correlations in One-Dimensional $1/|x-y|^2$ Percolation, Ising and Potts models}, Commun. Math. Phys. \textbf{118}, 303-336 (1988).
	
	\bibitem{LP}J. Littin and P. Picco. \textit{Quasi-additive estimates on the Hamiltonian for the one-dimensional long range Ising model}, J. Math. Phys.,  \textbf{58},  073301, (2017).
	
	\bibitem{Park1} Y.M. Park. \textit{Extension of Pirogov-Sinai Theory of Phase Transition to infinite range interactions. I. Cluster Expansion}, Comm. Math. Phys. \textbf{114}, 187-218, (1988).
	
	\bibitem{Park2} Y.M. Park. \textit{Extension of Pirogov-Sinai Theory of Phase Transition to infinite range interactions. II. Phase Diagram}, Comm. Math. Phys. \textbf{114}, 219-241, (1988).
	
	\bibitem{Pei} R. Peierls. \textit{On Ising's model of ferromagnetism}.
	Proc. Cambridge Philos. Soc. \textbf{32},
	477-481, (1936).
	
	\bibitem{Pfis} C.E. Pfister. \textit{Large Deviations and Phase separation in the two-dimensional Ising model}, Hel. Phys. Acta, \textbf{64}, 953-1054, (1991).
	
	\bibitem{Pirogov} S. A. Pirogov, Ya. G. Sinai. \textit{Phase diagrams of classical lattice systems}, Theoret. and Math. Phys., \textbf{25}, 1185-1192, (1975).
	
	\bibitem{Pres} E. Presutti. \textit{Scaling Limits in Statistical Mechanics and Microstructures in Continuum Mechanics}, Theoretical and Mathematical Physics, Springer, (2009).
	
	\bibitem{Firas} F. Rassoul-Agha, T. Sepp\"{a}l\"{a}inen. \textit{A Course on Large Deviations with an Introduction to Gibbs Measures}, Graduate Studies in Mathematics, AMS, (2015).
	
	\bibitem{Zaradnik} M. Zahradn\'{i}k. \textit{An alternate version of Pirogov-Sinai theory}, Commun. Math. Phys. \textbf{93}, 559-581, (1984).
	
\end{thebibliography}

\end{document}